\documentclass[12pt]{article}

\RequirePackage{fix-cm}

\usepackage{setspace}
\usepackage{indentfirst}

\usepackage{amsmath}
\usepackage{amsfonts}
\DeclareMathOperator{\E}{\mathbbm{E}}

\usepackage{graphicx}
\usepackage[normalem]{ulem}
\usepackage{titlesec}
\usepackage{titletoc}

\usepackage{mathtools}

\usepackage{booktabs}

\usepackage{ulem} 

\usepackage{amssymb}

\titleclass{\part}{top} 
\titleformat{\part}[display]
{\centering}
{}
{20pt} 
{\resizebox{!}{1cm}{\uline{\partname\ \thepart}}\\\vspace*{1em}\Huge} 
[\vspace*{4cm}\clearpage] 

\usepackage{amsmath,amssymb}

\usepackage{changepage}
\usepackage{caption}

\usepackage{textcomp,marvosym}

\usepackage{nameref}
\usepackage[hidelinks]{hyperref} 

\usepackage[right]{lineno}

\usepackage{microtype}

\usepackage[table]{xcolor}
\usepackage{color}
\usepackage{algorithm}
\usepackage{algpseudocode} 

\usepackage{amsthm}

\usepackage{thmtools}
\usepackage{thm-restate}

\newtheorem{example}{Example}

\newtheorem{defi}{Definition}
\newtheorem{assume}{Assumption}
\newtheorem{remark}{Remark}

\usepackage{array}

\newcolumntype{+}{!{\vrule width 2pt}}

\newlength\savedwidth

\usepackage{lastpage,fancyhdr,graphicx}
\usepackage{epstopdf}

\usepackage{amsfonts}
\usepackage{amsmath}
\usepackage{bbm}
\usepackage{dsfont}
\usepackage{bm}

\newcommand{\ind}{{\mathbbm{1}}}

\newcommand{\Prob}{\mathbbm{P}} 
\newcommand{\vect}[1]{\boldsymbol{#1}}

\DeclareMathOperator*{\argmax}{\arg\max}
\DeclareMathOperator*{\argmin}{\arg\min}

\newcommand{\syn}{\textrm{syn}}

\usepackage[numbers]{natbib}
\usepackage{PRIMEarxiv}

\usepackage[utf8]{inputenc} 
\usepackage[T1]{fontenc}    
\usepackage{url}            
\usepackage{booktabs}       
\usepackage{amsfonts}       
\usepackage{nicefrac}       
\usepackage{microtype}      
\usepackage{lipsum}
\usepackage{graphicx}       
\graphicspath{{media/}}     

\title{A Causal Inference Approach of Monosynapses from Spike Trains
}

\author{
  Zach Saccomano \\
  \textit{School of Neuroscience, Virginia Tech} \\
  \texttt{zachsaccomano@vt.edu} \\
  \And
  Sam McKenzie \\
  \textit{Health Science Center, University of New Mexico} \\
  \texttt{samckenzie@salud.unm.edu} \\
  \And
  Horacio G. Rotstein \\
  \textit{Federated Department of Biological Sciences,} \\
  \textit{New Jersey Institute of Technology \& Rutgers University} \\
  \And
  Asohan Amarasingham \\
  \textit{Department of Mathematics, The City College of NY} \\
  \textit{Depts. of Biology and Computer Science, The Graduate Center} \\
  \textit{City University of New York} \\
  \texttt{aamarasingham@ccny.cuny.edu} \\
}

\begin{document}
\maketitle

\begin{abstract}
Neuroscientists have worked on the problem of estimating synaptic properties, such as connectivity and strength, from simultaneously recorded spike trains since the 1960s. Recent years have seen renewed interest in the problem, coinciding with rapid advances in the technology of high-density neural recordings and optogenetics, which can be used to calibrate causal hypotheses about functional connectivity. Here, a rigorous causal inference framework for pairwise excitatory and inhibitory monosynaptic effects between spike trains is developed. Causal interactions are identified by separating spike interactions in pairwise spike trains by their timescales. Fast algorithms for computing accurate estimates of associated quantities are also developed. Through the lens of this framework, the link between biophysical parameters and statistical definitions of causality between spike trains is examined across a spectrum of dynamical systems simulations. In an idealized setting, we demonstrate a correspondence between the synaptic causal metric developed here and the probabilities of causation developed by \citet{tian2000probabilities}. Since the probabilities of causation are derived under distinct assumptions and include data from experimental randomization, this opens up the possibility of testing the synaptic inference framework's assumptions with juxtacellular or optogenetic stimulation. We simulate such an experiment with a biophysically detailed channelrhodopsin model and show that randomization is not achieved; strong confounding persists even with strong stimulations. A principal goal is to ask how carefully articulated causal assumptions might better inform the design of neural stimulation experiments and, in turn, support experimental tests of those assumptions. \color{black} 
\end{abstract}


\section{Introduction}

Various lines of experimental evidence suggest that, in some neuronal pairs,  monosynaptic input can reliably produce a postsynaptic spike response \textit{in vivo} \citep{English2017,Csicsvari1998} with a delay and precision that acts on millisecond timescales. Moreover, it appears that the most plausible explanation for corresponding observations of appropriately-timed millisecond-timescale correlations is the presence of a monosynaptic connection between the two cells. What is more, there is evidence that the magnitude of such fine-timescale correlations co-vary with the synapse's strength~\citep{jouhanneau2018single}. This suggests that a careful study of millisecond-timescale correlations in simultaneously-recorded spike trains might be a tool for studying synaptic dynamics during behavior \cite{Fujisawa2008}.

In practice, the hypothesis that monosynaptic effects can act on millisecond timescales is often incorporated into their analysis by a statistical formulation of a separation of timescale hypothesis. This formulation can be appreciated by looking at anecdotal examples of cross-correlograms (CCGs) from studies~\cite{English2017} that offer support for a causal interpretation by juxtacellular and optogenetic stimulation of putative presynaptic neurons \textit{in vivo} (see Figure~\ref{fig:neural_interventions}). Nevertheless, causal claims in highly connected systems ought to be treated delicately. It has been suggested that isolating fine timescale effects might be a way to sidestep such concerns~\citep{mehler2018lure,platkiewicz2021monosynaptic}. While many methods have been proposed for monosynaptic inference, relatively few have modeled causal relationships explicitly (but see ~\citep{lepperod2023inferring,platkiewicz2021monosynaptic}). Furthermore, many such methods operate on the CCG, and even under a timescale-separation assumption, the CCG is insufficient to identify synaptic properties even in quite simple models (Figure~\ref{fig:neural_interventions})~\citep{stevenson2008inferring}.

The primary focus of this study is to contribute to the development of robust and rigorous approaches to monosynaptic inference in which the causal inference is explicit. We develop a causal inference framework for monosynaptic interactions that is based on separation of timescale hypotheses that are robust to strong forms of nonstationarity in the background dynamics (the concern we have most heavily emphasized in the context of spike train analysis more generally \citep{Amarasingham2006,Amarasingham2012,Amarasingham2015}), among other forms of model misspecification. Unbiased estimators and confidence intervals for causal quantities are derived under rigorously-articulated statistical assumptions, and we develop accurate, efficient algorithms for computing these quantities. The performance of causal inference is then examined over broad parameter ranges in simulations of increasing complexity, ranging from point process models to adaptive exponential integrate-and-fire (AdEx) neuron models. 

We also use simulations to examine the correspondence between the causal metrics for monosynaptic interaction developed here and the {probabilities of causation,} as developed in \citet{tian2000probabilities} and elsewhere, which quantify the necessity and sufficiency of causation probabilistically. While the correspondence is studied in a setting that relies on strong idealizations, variations that are more finely tuned to experimental work that incorporates system-specific constraints might use an analogous correspondence to test the model's assumptions \textit{in vivo} or calibrate its free parameters via stimulation. Toward this goal, we use a biophysically detailed opsin model to simulate such an experiment \textit{in silico}. Using a theoretically motivated stimulation paradigm, we demonstrate that common input correlations might be difficult to disentangle from common causal influences with current experimental technologies, motivating future research on that point.

\section{Preliminary considerations and general architecture}
\label{sec:general_causal_model}

We begin with a general discussion of causal models to facilitate a uniform comparison between models and simulations instantiated at different levels of abstraction. As has been widely discussed, the key motivation for explicitly modeling causation is to distinguish association from causation. \, 
An intuitive model for doing so can be described by potential outcomes~\citep{neyman1923application,rubin1974,imbens2015,wasserman2013all}. We write $\vect Y^{(\vect X=\vect x)}$ as the `potential outcome' of the random variable $\vect Y$ if the variable $\vect X$ is `forced' to take the value $\vect x.$ The random 
variables $\vect X, \vect Y,$ as well as those of the form $\vect Y^{(\vect X=\vect x)}$ are presumed defined on a common probability space. We distinguish observational from experimental trials in this way. In experimental trials, the behavior of an agent external to the system (i.e., an agent that intervenes
{\it on} the system) is explicitly-modeled; in observational trials, there is no such intervention.  For example, in a drug efficacy trial, let $\{\vect X_k=0\}$ represent the event that patient $k$ takes the treatment and let $\{\vect X_k=1\}$ represent the event that patient $k$ takes a placebo. Then, if $\vect Y_k$ represents the measured outcome (mortality, for example) for patient $k$ in an {\it observational trial}, then $\vect Y_k^{(\vect X_k=1)}$ represents the measured outcome for patient $k$ in an {\it experimental trial}, such as a randomized control trial (RCT), in which patient $k$ has been assigned to take the treatment by a mechanism or agent external to the modeled system. Potential outcomes represent answers to questions of the form `\textit{What would happen if an external agent intervened on the system?}' and are used to define causal relations. In this case, the causal effect of the drug on the measured outcome, for patient $k$, is $\vect Y_k^{(\vect X_k=1)} - \vect Y_k^{(\vect X_k=0)}.$ It is commonly pointed out that the challenge of causal inference is that one of $\vect Y_k^{(\vect X_k=1)}$ and $\vect Y_k^{(\vect X_k=0)}$ is unobservable. The potential outcomes notation makes it straightforward to demonstrate that an RCT is designed to infer {\it average} causal effects across a population, e.g., ${\mathbb E} \left[ Y_R^{(\vect X_R=1)} - \vect Y_R^{(\vect X_R=0)} \right],$ where $R$ is a patient chosen in a simple random sample. (The latter demonstration assumes {\it consistency,} which is the assumption that the events $\{\vect Y_k=\vect y,\vect X_k=\vect x\}$  and $\{\vect Y_k^{(\vect X_k=\vect x)}=\vect y\}$ are identical.) 

A {\it constructive} way to model potential outcomes is by explicitly modeling interventions in terms of how a structural causal model~\citep{pearl2009} is simulated. In this approach, the structure of a system is modeled via the relationships among its variables, specified by a set of functions, as in a dynamical system. This set of functions can be put in correspondence with a directed graph by associating each variable with a vertex: a source vertex and a target vertex has a directed edge if one of the functions has the source in its domain and the target in its range. It is required that the directed graph is acyclic, which is equivalent to requiring that there is a consistent (sequential) method of simulating the system, whose ordering respects the graph. The {\it background} variables (noise variables) -- those variables whose corresponding vertices do not have incoming edges -- are instantiated as independent random variables and represent the influence of the world external to the system, in the absence of interventions.  We can then think of the simulation as a closed system. Random variables can be sampled by simulation. The background variables are sampled as noise terms. Probability propagates via the functional relationships to induce a joint probability distribution on the entire system of variables. This joint distribution specifies all probability distributions (conditional and marginal distributions) of interest. All such probability distributions can then, in principle, be estimated by simulation.  This is a more or less standard probabilistic point of view. The language of association is the language of conditional distributions. The association between random variables $\vect X$ and $\vect Y$ describes the distribution $P(\vect X|\vect Y);$ there is no association if $P(\vect X|\vect Y)=P(\vect X).$  

 We can describe {\it interventions} explicitly in the sequential simulation just identified. In this description, intervening on some variables means explicitly resetting their values before the functions that call them are evaluated (in the sequential method of simulation). These resets are the interventions; interventions model the action of agents external to the system. The random variables sampled in simulating the intervened system are potential outcomes. The $\vect Y^{(\vect X=k)}$ encodes the outcome for variable $\bf Y$ but in the {\it intervened} system in which $\vect X$ is reset to the value $k$ before its use in function calls. As before, the probabilities propagate via the functional relationships and the interventions to induce a new joint probability distribution on the entire system. This joint distribution specifies all probability distributions of interest in the intervened system. All such probability distributions can again, in principle, be estimated by simulation. Causal language can then be understood as a vocabulary for discussing how interventions modify probabilities of interest. Pearl \cite{pearl2009} uses the term $do(\cdot)$ to specify probability distributions for intervened systems. $P(\vect Y| do(\vect X=\vect x),\vect Z)$ represents the conditional distribution $P(\vect Y|\vect Z)$ when the system is intervened upon by assigning random variable $\vect X$ to the value $\vect x,$ where assigning is taken in the sense of  `resetting' above. Thus $P(\vect Y| do(\vect X=\vect x),\vect Z)$ is another way of writing $P(\vect Y^{(\vect X=x)}  | \vect Z^{(\vect X=x)}  )$. In what follows, we use either notation freely, for convenience.

\begin{figure}[!htbp]
    \centering
\includegraphics[width=.7\textwidth]{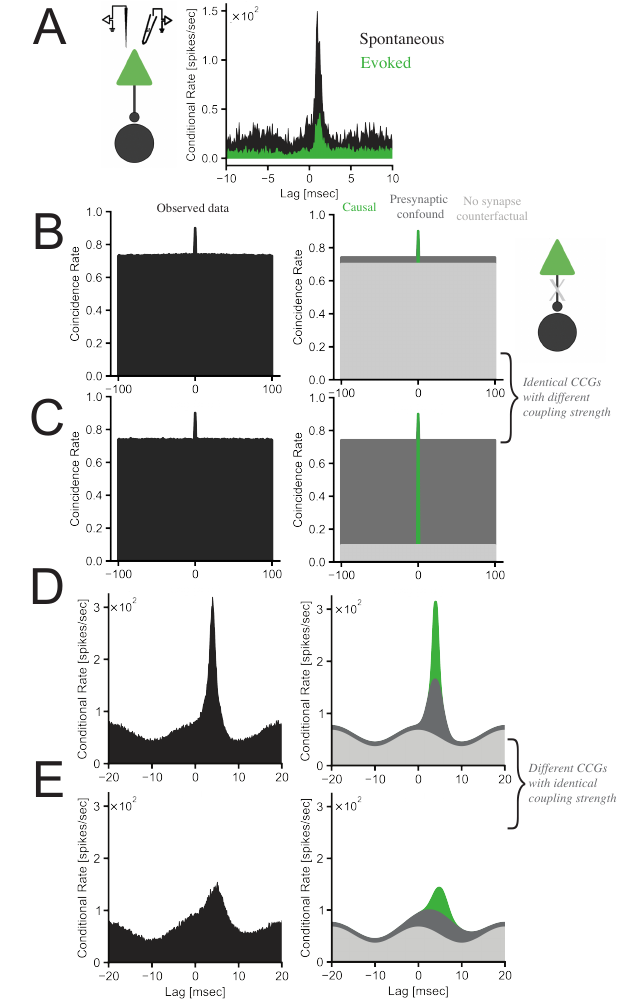}
    \caption[Neural interventions and confounding]{{\bf Toy examples of confounding in monosynaptic interactions.} Simulations for examples are explained in detail in Appendix~\ref{app:ccg_exam}. \textbf{A:} A CCG from a hippocampal pyramidal cell and interneuron \textit{in vivo} from the study of \citet{English2017}. The pyramidal neuron (hypothesized to be presynaptic) 
    spontaneously spikes (black) or spikes in response to experimental juxtacellular stimulation (green). Data like these motivate timescale separation assumptions. \textbf{B: } CCG from {Example~\ref{example:simple_ccg_degeneracy}} for \textbf{Situation A} where $\epsilon_A = 0.2, \lambda_{R,A} = 0.2, \lambda_{T,A} = 0.7$. Left panel is a simulation, right panel is analytic for all rows. \textbf{C: } \textbf{Situation B} in {Example~\ref{example:simple_ccg_degeneracy}} where $\epsilon_B = 0.8, \lambda_{R,B} = 0.8, \lambda_{T,B} = 0.1$. Note the observed CCGs are the same in \textbf{B} and \textbf{C} but the causal stories generating the data are quite different. \textbf{D: } \textbf{Situation A} in {Example~\ref{example:complex_ccg_degeneracy}} \textbf{E: } \textbf{Situation B} in {Example~\ref{example:complex_ccg_degeneracy}}. For \textbf{D-E} the parameters are $\sigma_A^2 = 10$ ms, $\sigma_B^2 = 90$ ms, $\sigma_s^2 = 2.5$ ms, $d = 4$ ms, $\alpha = 0.005$, $\omega = 20$ Hz. The observed CCGs in the left panels of \textbf{D} and \textbf{E} are quite different but in the right panels the number of causal spikes (the areas of the green regions) are identical and the presynaptic trains, and thus their ACGs, are also identical.} 
\label{fig:neural_interventions}
\end{figure}

Where does this language  -- in which interventions by agents are explicitly modeled -- improve upon more familiar statistical modeling? A textbook example is ``Simpson's Paradox'', a phenomenon where the observed statistical association between variables in a population is the opposite of that observed within each subgroup of a partition of the population (see~\citep{wasserman2013all,pearl2016causal} for more explanation). In neurophysiology, perhaps the most familiar object used in the laboratory for synaptic inference is the cross-correlogram (CCG). To motivate the framework of this article, we construct examples where the CCG and presynaptic ACG~\cite {spivak2022deconvolution,stevenson2023circumstantial} are insufficient for causal inference. The CCG hides information about time-dependent background correlations and presynaptic bursts that produce associations in the CCG that are not causal. The situation is analogous to \textit{Simpson's Paradox} where statistical associations must be isolated in strata of confounders that have specific properties. The CCG collapses over these strata of the confounders. {Figure~\ref{fig:neural_interventions}} illustrates these ideas in some point process examples, with a detailed mathematical explanation of the simulations given in Appendix~\ref{app:ccg_exam}.

\begin{table}[htp]
\centering
\caption{Description of symbols and notation in monosynaptic causal inference model}
\label{tab:symbol_descriptions}
\begin{tabular}{ll}
\toprule
\textbf{Symbol} & \textbf{Description} \\
\midrule
$(A_k)_{k \in \mathbbm{Z}^*}$ & An ordered sequence $(A_0,A_1,...)$ \\
$\vect A$ & A set of numbers $\{A_1,A_2,...\}$ \\
$\vec{\vect A}$ & A matrix \\
$\vec{a}$ & A vector \\
$|A|$ & The cardinality of a set $A$ or $A$'s absolute values if $A$ is a scalar \\
$\mathbbm{1}\{A\}$ & The indicator function of the event A \\
$\vect A^{[n]}$ & The set of all subsets of cardinality $n$ from a set $\vect A$ \\
$\vect X, \vect Y$ & Generic sets of points (e.g., spike times) often reused \\
$N_A(\vect X)$  & Number of spikes from $\vect X$ in the temporal regions A  \\
$\gamma(t)$  & The $\Delta$ coarse temporal interval containing $t$  \\
$\vect \Gamma(\vect X)$  & An abbreviation for $(N_{\gamma(\Delta k)}(\vect X))_{k \in \mathbbm{Z}^*}$  \\
$S(\vect X,\delta,\tau)$  & The union of all $\delta$ length intervals centered around the $\tau$ shifted elements of $\vect X$ \\
$\delta, \tau, \Delta$ & The synaptic timescale and lag, and background timescale (model parameters) \\
$\delta_d$ & The Dirac delta function \\
$\vect Y^{(\vect x)}$  & The potential outcome of a spike train $\vect Y$ if a spike train $\vect X$ is forced to be $\vect x$ \\
$\vect R, \vect B, \vect I, \vect T$  & The reference, background, interaction, and target events, respectively \\
$\theta_{syn}$  & The number of interactions caused by a synapse \\
$q(\vect R, x)$  & The conditional probability some $x \in \vect B$ will be found in $S(\vect R,\delta,\tau)$ \\
$\vect q(\vect R, \vect T^{(\vect R)})$  & The collection $(q(\vect R, x))_{x \in \vect T^{(\vect R)}}$  \\
$\vect G$  & The set of reference spikes with no target spike found in their interaction regions  \\
$\vect K$  & An index set for $\vect T$ \\
$\vect J$  & The indices of $\vect K$ labeling which elements of $\vect T$ equal $\vect B$ when $\theta_{syn} \geq 0$\\
$\vect L$ & The indices of $\vect K$ labeling synchronous elements of $\vect T$ when $\theta_{syn} \geq 0$ \\
$\vect U_h$, $\vect V_h$ & Indices of $\vect L$ corresponding to two limiting cases for the hypothesis $\theta_{syn} = h \geq 0$ \\
$\alpha$ & One minus the confidence level \\
$\vect J_h^-$, $\vect J_h^+$  & Hypotheses for $\vect J$ given the hypothesis $\theta_{syn} = h \geq 0$ \\
$\vect Z$ & Candidate points that might be near inhibitory events $\vect T^{(\emptyset)} \setminus \vect T^{(\vect R)}$ \\
$\Tilde{\vect K}$ & An index set for $\vect Z$ \\
$\Tilde{\vect J}$ & The indices in $\Tilde{\vect K}$ for points in $\vect Z$ near $\vect T^{(\emptyset)} \setminus \vect T^{(\vect R)}$ \\
$\Tilde{\vect J}_h^-$, $\Tilde{\vect J}_h^+$  & Hypotheses for $\Tilde{\vect J}$ given the hypothesis $\theta_{syn} = h$ \\
$\chi(\ell,\vect X, \vect Y)$ & The sample cross-correlation function between spike trains $\vect X$ and $\vect Y$ \\
$\mathcal{C}(\vect R, \vect T, \alpha)$ & $1-\alpha$ confidence interval for $\theta_{syn} \geq 0$ \\
\end{tabular}
\end{table}

\newpage
\section{Monosynaptic causal inference model}
\subsection{Formulation for primary model}
\label{sec:mono_causal_infer_model}

Let $\vect X$ be a finite set of spike times for an experiment of duration $D$. For any $A = \cup_i [a_i,b_i)$, let
$N_A(\vect X) \coloneqq |\vect X \cap A|$, termed the \textit{increment} of the point process $\vect X$ in $A$~\citep{Kass2005,daley2008introduction}. Assuming the time origin is randomized in the experimental sense, implicitly define a partition of $\mathbbm{R}^+$ into  $\Delta$ length intervals with a function that for any time point $t \in \mathbbm{R}^+$ retrieves the unique coarse interval containing $t$,
\begin{equation}
    \gamma(t) \coloneqq \bigg\{ I \in \{[k\Delta,k\Delta+\Delta) : k \in \mathbbm{Z}^*\} : t \in I \bigg\}.
\end{equation}
Thus the $k$-th interval will frequently be written by taking $t$ at its left endpoint, $\gamma(k\Delta)$. However, in future sections, it will be equally convenient to access interval $k$ by letting $t$ equal any spike time in interval $k$. Abbreviate the sequence of spike counts for some spike times $\vect X$ in the coarse intervals with the special symbol
\begin{equation}
    \vect \Gamma(\vect X) \coloneqq (N_{\gamma(k\Delta)}(\vect X))_{k \in \mathbbm{Z}^*} = (N_{[0,\Delta)}(\vect X), N_{[\Delta,2\Delta)}(\vect X),...)
\end{equation}
\noindent and for any $\vect X$ denote the union of all $\delta$ length intervals centered around the $\tau$ shifted elements of $\vect X$ as 
\begin{align}
S(\vect X,\delta,\tau) \coloneqq \bigcup_{x \in \vect X} \left\{s : |x + \tau - s| \leq \frac{\delta}{2}\right\}    
\label{eq:synch_region}
\end{align}
\noindent where the second two arguments will often be suppressed from the notation when the context is clear, i.e., $S(\vect X)$. Let $\vect Y^{(\vect x)}$ be the potential outcome of a spike train $\vect Y$ if a spike train $\vect X$ is forced to spike at a set of times $\vect x$ with otherwise fixed background conditions. That is, introduce the counterfactual notion that $\vect Y^{(\vect x)}$ \textit{would have been} the set of times $\vect Y$ spiked if $\vect X$ had been $\vect x$~\citep{imbens2015}. As previewed above, we will work with a reference spike train, $\vect R$, and a target train, $\vect T$, and the scientific question is to ask if $\vect R$ acts on $\vect T$ with a monosynapse. Further suppose the target train is constructed from latent point processes $\vect B$ and $\vect I$, termed background events and interactions, respectively. A causal model is then defined with the deterministic relation,
\begin{align}
&\forall \vect r, \vect T^{(\vect r)} \coloneqq \bigg\{\begin{array}{lr}
        \vect B \cup \big(\vect I^{(\vect r)} \setminus \cup_{r \in \vect r} \{S(r) : N_{S(r)}(\vect B) > 0\}\big), \text{ excitatory model }\\
        \vect B \setminus \cup_{r \in \vect r} \{S(r) : N_{S(r)}(\vect I^{(\vect r)}) > 0\}, \text{ inhibitory model } 
        \end{array} \label{eq:synon}
\end{align}

where $\vect T^{(\vect r)}$ references the intervention $do(\vect R = \vect r)$. Notice, by construction, the background events $\vect B$ are invariant to any action $do(\vect R = \vect r)$. While one might argue this is a strong simplification, it is appropriate to compare it to the assumption that smooth features in the CCG are non-causal~\citep{kobayashi2019} which, for Poisson-based models, is necessarily a subset of the simplification just made (Figure~{\ref{fig:neural_interventions}C-D})\color{black}. We will be concerned with estimation of the parameter $\theta_{syn} = N_{S(\vect R)}(\vect T)  - N_{S(\vect R)}(\vect T^{(\emptyset)})$ where $\theta_{syn} > 0$ for excitatory interactions, $\theta_{syn} < 0$ for inhibitory interactions, and $\theta_{syn} = 0$ for non-interacting neurons. In the following, it will be useful to define,
\begin{equation}
    q(\vect R, x)  \coloneqq \frac{1}{\Delta} \int_{t \in \mathbbm{R}^+} \mathbbm{1}\{t \in S(\vect R) \cap \gamma(x)\} \; \; dt.
\end{equation}

That is, $q(\vect R, x)$ is the proportion of times $t \in \gamma(x)$ that are within a distance $\delta/2$ of a point in $\{r + \tau : r \in \vect R\}$. Constructing confidence intervals for $\theta_{syn}$ will require some additional notation. First, let us set up objects that will be used for an exact excitatory confidence interval. Denoting $T_k=f_0(k)$, fix any bijective mapping  $f_0: \vect K \mapsto \vect T^{(\vect R)}$ that satisfies
\begin{align}
\label{eq:rankord}
\centering
    &q(\vect R, T_{1}) \leq q(\vect R, T_{2}) \leq ... \leq q(\vect R, T_{|\vect T|}).
\end{align}

\noindent We will write
\begin{align} 
\label{eq:qdef}
\centering
    &\vect q( \vect R, \vect T^{(\vect R)} ) = ( q(\vect R, T_{1}),  q(\vect R, T_{2}), ...,q(\vect R, T_{|\vect T|})).
\end{align}

Let $\vect J$ denote the subset of $\vect K$ indexing the true background events, $\vect B$, in the sense that $\vect J$ satisfies $\{T_j : j \in \vect J \text{ and } \vect J \subseteq \vect K\} = \vect B$.

With this notation fresh in mind, we define a background model using the principle of conditional uniformity, which has been motivated and developed as a canonical assumption in previous work~\citep{amarasingham2011conditional,Harrison2014}. For our purposes here, the following technical definition is sufficient (see~\cite{daley2008introduction} for more on point processes and their characterization).

\begin{defi}
\textbf{Conditionally uniform point process}: Define $g(\vect Y,A) = | \vect Y \cap \gamma( \inf A ) |,$ where $\vect Y$ is a point process and $A$ is a subset of $\mathbb{R}.$ A point process \( \vect Y \) is {\it conditionally uniform, conditioned on \( \vect \Gamma(\vect Y) \) and \( \vect X\) }, if
\begin{equation}
\label{eq:cond-unif}
\mathbbm{P}\left( \cap_{k=1}^m \{ | \vect Y \cap A_k | = n_k \} | \vect \Gamma(\vect Y), \vect X \right) = \prod_{k=1}^m \binom{ g(\vect Y,A_k) }{ n_k} \left( \frac{ |A_k| }{\Delta} \right)^{n_k} \left( 1 - \frac{ |A_k| }{\Delta} \right)^{g(\vect Y,A_k) - n_k}, 
\end{equation}
if $n_k \leq g(\vect Y,A_k)$ for all $k \in \{1,2,..., m\},$ for any disjoint finite collection \( A_1, A_2, ..., A_m \) of subsets of \(\mathbb{R}\) that satisfies: i) for each \( j, A_j \subseteq [ k_j\Delta, k_j\Delta + \Delta) \) for some integer \( k_j \) and ii) \( \gamma( \inf A_{j_1} ) \not= \gamma( \inf A_{j_2} )  \) whenever \( j_1 \not= j_2.\)
\end{defi}

A common way of modeling point processes is with \textit{conditional intensity functions}~\citep{Kass2005}. While the formulation just outlined does not make use of them, later we will simulate from conditional intensity function models to demonstrate this formulation is compatible. In continuous time, the conditional intensity function $\lambda_{\vect X}(t)$ for a point process $\vect X$ is, $\lambda_{\vect X}(t) \coloneqq \lim_{\Delta t \rightarrow 0} \E[N_{[t,t+\Delta t)}(\vect X) | \mathcal{H}_t]/\Delta t$ where $\mathcal{H}_t$ is the history of the system prior to time $t$.

 \begin{remark}
 In neuroscience, the term ``rate'' might refer to one of several ideas~\cite {Amarasingham2015}. In the current work, we use the word rate with regard to samples of a conditional intensity function and highlight the normalization when used.
 \end{remark}

\subsection{Assumptions}

\begin{assume} \label{as:4}
\textbf{Conditional uniformity:} $\vect B$ is a conditionally uniform point process, conditioned on $\vect \Gamma(\vect B)$ and $\vect J$.
\end{assume}

\begin{assume} \label{as:2}
\textbf{Timescale separation: } For some $\tau$, $\delta$, and $\Delta, \vect I^{(\vect r)} \subset S(\vect r)$, for all $\vect r$, where $\delta < \Delta$. ($\tau$, $\delta$, and $\Delta$ are model parameters.)
\end{assume}

\begin{assume} \label{as:3}
\textbf{Positivity: } $0 < q(\vect R, k\Delta) < 1$, for all $k \in \mathbbm{Z}^*$.
\end{assume}

\begin{assume} \label{as:1}
\textbf{Consistency: } $\vect I = \vect I^{(\vect R)} \text{ and } \vect T = \vect T^{(\vect R)} $.
\end{assume}

Assumption $i$ will be abbreviated as $\mathcal{A}.i$. There have been debates about whether $\mathcal{A}.$\ref{as:1} (consistency) is an assumption or axiom of causal inference \citep{pearl2010brief,vanderweele2009concerning,cole2009consistency}. Here, we take it as an assumption in the sense of highlighting where it is invoked or self-evident. Similarly, one can view $\mathcal{A}.$\ref{as:3} (positivity) as an identifiability condition, and in our case, its validity can be determined with observational data. For this reason, one could simply define $\theta_{syn}$ in terms of regions of an experiment that provide identifiable causal information. However, following the causal inference literature~\citep{hernan2010causal} and for full conceptual clarity, we make it an assumption which more easily accommodates an explanation of both perspectives, leaving scientists to make their own judgment within the context of specific questions. In particular, the assumption interrelates with various other issues, including choosing free parameters, which will be discussed at length. For this purpose, the following will be useful.

\begin{defi} 
\textbf{Synchrony saturation:} Synchrony saturation refers to an observation of the model where $\mathcal{A}.$\ref{as:4} (conditional uniformity) and $\mathcal{A}.$\ref{as:2} (timescale separation) are true but $\mathcal{A}.$\ref{as:3} (positivity) is violated. That is, we say an observation is synchrony saturated if there is an interval identified by $k \in \mathbbm{Z}^*$ such that $q(\vect R, k\Delta) = 1$.
\end{defi}

The primary motivation for $\mathcal{A}.$\ref{as:2} (timescale separation) is empirical~\citep{Csicsvari1998,Fujisawa2008,English2017,jouhanneau2018single} although this assumption will be investigated in simulations of dynamical systems in later sections. Finally, $\mathcal{A}.$\ref{as:4} (conditional uniformity) is motivated by the observation that \textit{in vivo} spike trains are nonstationary~\citep{shinomoto1999ornstein,Softky1993}, and likely rapidly-varying. Hence, distinct points in time cannot be averaged to estimate conditional intensity functions or their variants, such as the cross-correlation function~\cite {Amarasingham2015,amarasingham2011conditional}, a matter made worse by confounding. As discussed in past work~\citep{Harrison2014}, the particular use of \textit{uniformity} is motivated by the fact that the uniform distribution is the maximum entropy distribution on a finite interval.

\subsection{Point estimation}

Perhaps the key task of causal inference is to identify confounders and adjust for them. In the assumptions just put forth, it is conceived that the processes $\vect B$, $\vect R$, and $\vect I^{(\vect R)}$ may have non-trivial correlations that confound $\theta_{syn}$ on a $\Delta$ timescale. In causal inference, adjustment often ensues by stratifying the probability of the outcome variable conditioned on the treatment variable into different levels of the confounder. Similarly, here, the key to estimation will be to stratify time into $\Delta$ length neighborhoods and perform statistical adjustments locally (i.e., in time) via conditioning on the spike counts in those intervals. Note that a quite different approach would be to use the CCG for estimation where the spiking activity has already been averaged across levels of confounding. Figure \ref{fig:neural_interventions}, and its associated examples, essentially demonstrated that the decomposition of the CCG into causal parts is an ill-posed problem since information about confounding is hidden after averaging. Notice in the formulation section no assumptions were made about $\vect R$ and no assumptions were made about $\vect I^{(\vect R)}$ except for $\mathcal{A}.$\ref{as:2} (timescale separation). In the following theorem, an unbiased estimator is provided for $\theta_{syn}$. This precise deconfounding of the synaptic effect comes at the cost of not modeling the time-dependent shape of the synaptic gain onto the postsynaptic neuron. Instead, $\mathcal{A}.$\ref{as:2} (timescale separation) simply requires interactions $\vect I^{(\vect R)}$ to be a subset of $S(\vect R)$. This does not require that no shape exists, it is simply not inside the model. Another potential source of confusion is that the idealization that there exist two processes $\vect B$ and $\vect I^{(\vect R)}$ does not mean there are two levels of synaptic efficacy. Since the events $\vect B$ are invariant under all the actions $do(\vect R = \vect r)$, they indeed have zero effective synaptic weight. However, every event in $\vect I^{(\vect R)}$ may be generated from a different state-dependent effective synaptic weight. That is, $\mathcal{A}.$\ref{as:2} (timescale separation) is a statement about timescale, and no assumptions about synaptic gain were made or its dependence on other factors in the model. 

\begin{restatable}{theorem}{mytheoremone}
\label{thm:my_theorem_1}
Under $\mathcal{A}.$\ref{as:4} (conditional uniformity), $\mathcal{A}.$\ref{as:2} (timescale separation), $\mathcal{A}.$\ref{as:3} (positivity) and $\mathcal{A}.$\ref{as:1} (consistency) an unbiased point estimate of $\theta_{syn}$ in the excitatory and inhibitory models Eq.~\eqref{eq:synon} is given by one expression,
\begin{align}
\hat{\theta}_{syn} = \sum_{k \in \mathbbm{Z}^*}\frac{N_{\gamma(k\Delta) \cap S(\vect R)}\big(\vect T\big) -  q(\vect R, k\Delta) N_{\gamma(k\Delta)}\big(\vect T\big)}{1- q(\vect R, k\Delta)}.\label{eq:pointestimator}
\end{align}
\end{restatable}

\textbf{Proof Idea:} The observed (confounded) synchrony in $\Delta$-length temporal intervals can be expressed as a function of the hidden variables for each outcome. An appropriate conditional expectation yields calculations that isolate the causal effect as a function of observational data. Linearity of expectation across temporal intervals then recovers $\theta_{syn}$ (see Appendix \ref{app:proof}).

Under violations of the identifiability assumption $\mathcal{A}.$\ref{as:3} (positivity), one can easily salvage estimates from segments of the observation from which causal information is available.

\begin{restatable}{coro}{mycoroone}
\label{thm:my_coro1}
Suppose synchrony saturation occurs such that $q(\vect R, k \Delta) = 1$ where $k \in X, X \subseteq \mathbbm{Z}^*$. Define the set, $\nu = \cup_{k \in X} \gamma(k\Delta)$ and the parameter $\theta_{syn}' = N_{S(\vect R) \cap \nu}(\vect T)  - N_{S(\vect R) \cap \nu}(\vect T^{(\emptyset)})$. Then,

\begin{align}
\hat{\theta}_{syn}' = \sum_{k \in \mathbbm{Z}^* \setminus X}\frac{N_{\gamma(k\Delta) \cap S(\vect R)}\big(\vect T\big) -  q(\vect R, k\Delta) N_{\gamma(k\Delta)}\big(\vect T\big)}{1- q(\vect R, k\Delta)} 
\end{align}

\noindent is unbiased under $\mathcal{A}.$\ref{as:4} (conditional uniformity), $\mathcal{A}.$\ref{as:2} (timescale separation), and $\mathcal{A}.$\ref{as:1} (consistency).
\end{restatable}

\textbf{Proof Idea:} The proof (not shown) is exactly as before while highlighting the use of linearity of expectation mentioned in the previous proof idea.

\subsection{Confidence intervals}
\label{sec:monosynaptic_confidence_intervals}
The intuition behind the confidence intervals proposed here can be understood by explaining a naive algorithm for computing them. The algorithm's task is to explain the monosynaptic synchrony in a spike train pair in terms of the model. Consider the classical technique of obtaining a confidence interval by inverting a hypothesis test~\citep{casella2002statistical}. Intuitively, a confidence interval for $\theta_{syn}$ is the set of hypotheses for which we fail to reject the null hypothesis $H_0: \theta_{syn} = j$. A naive algorithm for calculating this interval would be to start with the hypothesis $H_0: \theta_{syn} = 0$. In this case, we assume all the observed spikes in the target train arise from the process $\vect B$. Taking monosynaptic synchrony as our test statistic, the test might reject the null hypothesis that is placed under the supposition that all spikes are non-causal and thus conditionally uniform ($\mathcal{A}.$\ref{as:4}). In that case, for an excitatory interval, we proceed to conduct more hypothesis tests $j=(1,2,3,...)$. Computationally, one can imagine that before each new hypothesis test, we delete synchronous target spikes from the target train, subtract away their contribution to the test statistic, and calculate the null distribution of the test statistic under the supposition that the remaining data are conditionally uniform ($\mathcal{A}.$\ref{as:4}). We continue this process until we fail to reject the null (i.e. until the observed synchrony is explained). 

The key question is, which spikes should be iteratively deleted from the target spike train in this naive algorithm? As will be shown analytically, we want to iteratively select synchronous target spikes that minimize or maximize the change in the tail probability of the test statistic. Under the model's assumption, this corresponds to removing synchronous target spikes that occur when the reference neuron's spike counts are either lowest or highest. Intuitively, suppose confounding events, $\vect B$, tend to occur when the reference neuron's firing is highest; the confounding synchrony will be maximal. In that case, $\theta_{syn}$ needs to be minimal to explain the synchrony. At the other extreme, if confounding events, $\vect B$, tend to occur when the reference neuron's firing is lowest, the confounding synchrony will be minimized, and thus $\theta_{syn}$ needs to be maximal to explain the synchrony. This is how $\theta_{syn}$ is bounded. We will now proceed to a more formal explanation of these intervals and prove they are exact. While the naive algorithm just described was for the purpose of intuition, a more sophisticated algorithm will also be developed for implementation later.

\subsubsection{Formulation and derivation for exact excitatory confidence interval}
\label{sec:excitatory_ci}


Define $\vect L \coloneqq \{ l \in \vect K : T_l \in S(\vect R) \}$ and let $\vect U_h$ be any set that satisfies $ \vect U_h  \subset \vect L$ such that $|\vect U_h| = N_{S(\vect R)}(\vect T)-h$ and $\max_i \{q(\vect R, T_i) : i \in  \vect U_h\} \leq \min_i \{q(\vect R, T_i) : i \in  \vect L \setminus \vect U_h\}.$ Similarly, let $\vect V_h$ be any set that satisfies $ \vect V_h  \subset \vect L$ such that $|\vect V_h| = N_{S(\vect R)}(\vect T)-h$ and $\min_i \{q(\vect R, T_i) : i \in  \vect V_h\} \geq \max_i \{q(\vect R, T_i) : i \in  \vect L \setminus \vect V_h\}$. $\vect U_h$ and $\vect V_h$ identify $N_{S(\vect R)}(\vect T)-h$ indices (specified by $f_0$) of the synchronous target spikes with the smallest and largest $q(\vect R, \cdot)$ values, respectively.  Then define
\begin{align}
    \label{eq:rankord1}
     \vect J^{-}_h &:=
       \vect U_h \cup  (\vect K \setminus \vect L)  \\
      \vect J^{+}_h &:= 
        \vect V_h \cup  (\vect K \setminus \vect L) .
    \label{eq:rankord2}
\end{align}

The conditional \textit{pmf} for $N_{S(\vect R)} (\vect T) - \theta_{syn} = N_{S(\vect R)} (\vect B),$ conditioned on $\vect q( \vect R, \vect T )$ and $\vect J$, is

\begin{equation} \label{bernoulli_eq}
\mathbbm{P}\bigg(N_{S(\vect R)} (\vect B) = n \bigg| \vect q( \vect R, \vect T ), \vect J \bigg) = \sum_{Q\in \vect J^{[\text{n}]}} \prod_{i \in Q} q(\vect R, T_i) \prod_{k\in {\vect J \backslash Q}}(1-q(\vect R, T_k)).
\end{equation}

Let $c^{-}( \vect q( \vect r, \vect t ),\vect j)$ specify a lower (conditional) critical threshold for $N_{S(\vect R)} (\vect B)$ in the sense,
\begin{align}
    \label{eq:critical_thres}
   c^{-}(\vect q( \vect r, \vect t ),\vect j) &\coloneqq \max_k \bigg \{ k : \mathbbm{P}\bigg(N_{S(\vect R)} (\vect B)  \leq  k \bigg | \vect q( \vect R, \vect T) = \vect q( \vect r, \vect t), \vect J = \vect j \bigg) \leq \alpha/2 \bigg \} \\
   &= \max_k \bigg\{ k : \sum_{n=0}^{k} \sum_{Q\in \vect j^{[\text{n}]}}^{} \prod_{i \in Q}^{} q(\vect r, t_i) \prod_{m\in (\vect j \setminus Q)} (1-q(\vect r, t_m)) \leq \alpha/2 \bigg\}.
\end{align}

In the same sense, let $c^{+}( \vect q( \vect r, \vect t ),\vect j)$ specify an upper (conditional) critical threshold for $N_{S(\vect R)} (\vect B)$,
\begin{align}
   c^{+}(\vect q( \vect r, \vect t ),\vect j) &\coloneqq \min_k \bigg \{ k : \mathbbm{P}\bigg(N_{S(\vect R)} (\vect B)  \geq  k \bigg | \vect q( \vect R, \vect T) = \vect q( \vect r, \vect t), \vect J = \vect j \bigg) \leq \alpha/2 \bigg \}. 
\end{align}

\noindent We now develop confidence intervals for $\theta_{syn}$. 

\begin{restatable}{lemma}{lemmaone}
\label{argmin-lemma}
Let
\begin{equation}
\mathcal{D}(\vect R, \vect T) = \bigg\{ \vect j : |\vect j| = |\vect J|, (\vect K \setminus \vect L) \subseteq \vect j   \bigg\}.
\end{equation}
Abbreviate $\mathcal{D}(\vect R, \vect T)$ as $\mathcal{D}$. Under $\mathcal{A}.$\ref{as:4} (conditional uniformity)
\begin{equation}
    \vect J^{-}_{\theta_{syn}} \in \argmin_{\vect j \in \mathcal{D}} \bigg \{ c^{-}( \vect q( \vect r, \vect t ),\vect j) \bigg \}
\label{eq:l_minus}
\end{equation}
and 
\begin{equation}
    \vect J^{+}_{\theta_{syn}} \in \argmax_{\vect j \in \mathcal{D}} \bigg \{ c^{+}( \vect q( \vect r, \vect t ),\vect j) \bigg \}.
\label{eq:argmax_plus}
\end{equation}
\end{restatable}

\textbf{Proof Idea:} The \textit{cdf} corresponding to the critical regions can be explicitly differentiated with respect to the labeling. Proving that $\vect J^{-}_{\theta_{syn}}$ minimizes Eq.~\eqref{eq:l_minus} follows from contradiction. Eq.~\eqref{eq:argmax_plus} is shown in the same way (see Appendix \ref{app:proof}).

An elementary idea embedded in Lemma \ref{argmin-lemma} above is that if $X_1, X_2, ..., X_n$ are independent Bernoulli random variables with parameters $q_1, q_2, ..., q_n$, respectively, and those parameters satisfy $q_i \geq r_i,$ for all $i \in \{1,2,...,n\}$, then $\mathbbm{P}( \sum_{i=1}^N X_i \leq c)$ is a lower bound for the same tail probability derived from a sum of independent Bernoulli random variables with parameters $r_1, r_2, ..., r_n, $ respectively. All of these Bernoulli random variables are presumed to be mutually independent. In fact, there is a more general result ~\citep{amarasingham_diss} which implies that the confidence intervals developed are robust to certain violations of $\mathcal{A}.$\ref{as:4} (conditional uniformity).

\begin{restatable}{lemma}{lemmatwo}
Denote $X_1, X_2, ..., X_{i-1}, X_{i+1}, ..., X_n$ as $_i X.$ Then suppose $X_1, X_2, ..., X_n$ are Bernoulli random variables ({\rm \bf not} necessarily independent) satisfying
\begin{equation}
    \mathbbm{P}( X_i=1 | _i X, Z ) \geq p_i(Z), 
\end{equation}
for all $i \in \{1,2, ..., n\},$ and for some random variable $Z$. Then
\begin{equation}
    \mathbbm{P}\left( \sum_{i=1}^n X_i \leq k \bigg| Z \right) \leq \mathbbm{P} \left( \sum_{i=1}^n Y_i \leq k \right), \forall k
\end{equation}
if $Y_1, Y_2, ..., Y_n$ are conditionally independent Bernoulli random variables, conditioned on $Z$, with parameters $p_1(Z), p_2(Z)$, $\ldots, p_n(Z)$ respectively, and $X_{1:n}$ and $Y_{1:n}$ are conditionally independent, conditioned on $Z$.
\end{restatable}

\textbf{Proof Idea:} The proof is by induction (see Appendix \ref{app:proof}).

First we first demonstrate that an exact hypothesis test for $H_0: \theta_{syn} = h$ follows from the previous results. 

\begin{restatable}{prop}{propone}
\label{prop-critregion}
Under $\mathcal{A}.$\ref{as:4} (conditional uniformity), $\mathcal{A}.$\ref{as:2} (timescale separation), and  $\mathcal{A}.$\ref{as:1} (consistency) $\{ N_{S(\vect R)}(\vect T) - h \leq c^{-}( \vect q( \vect R, \vect T), \vect J^{-}_{h}) \}$ and $\{ N_{S(\vect R)}(\vect T) - h \geq c^{+}( \vect q( \vect R, \vect T), \vect J^{+}_{h}) \}$ are $\alpha/2$-level critical region for all $\mathbbm{P}$ in $H_0: \theta_{syn}=h.$ That is, for all $\mathbbm{P}$ in $H_0: \theta_{syn}=h,$

\begin{equation} \label{left-tail-bd}
    \mathbbm{P}\left(  N_{S(\vect R)}(\vect T) - h \leq c^{-}( \vect q( \vect R, \vect T), \vect J^{-}_{h} ) \right) \leq \alpha/2
\end{equation}
and
 \begin{equation} \label{right-tail-bd}
    \mathbbm{P} \left(  N_{S(\vect R)}(\vect T) - h \geq c^{+}( \vect q( \vect R, \vect T), \vect J^{+}_{h} ) \right) \leq \alpha/2.
\end{equation}
\end{restatable}

\textbf{Proof:} Appendix \ref{app:proof}.

\noindent Finally, an exact confidence interval is constructed by inverting the hypothesis tests established in Proposition \ref{prop-critregion}.

\begin{restatable}{theorem}{theoremtwo}[Confidence interval for $\theta_{syn}$.]
Define
\begin{equation}
\mathcal{C}(\vect R, \vect T, \alpha) := \bigg\{ h: c^{-}( \vect q( \vect R, \vect T ), \vect J^-_h) \leq N_{S(\vect R)}(\vect T) - h \leq c^{+}( \vect q( \vect R, \vect T ), \vect J^+_h) \bigg\}.
\end{equation}

Then under $\mathcal{A}.$\ref{as:4} (conditional uniformity), $\mathcal{A}.$\ref{as:2} (timescale separation), and $\mathcal{A}.$\ref{as:1} (consistency) $\mathcal{C}(\vect R, \vect T, \alpha)$ is a $(1-\alpha)$-level confidence interval for $\theta_{syn} \geq 0.$ That is,

\begin{equation} \label{eq:CI}
\mathbbm{P}(\theta_{syn} \in \mathcal{C}(\vect R, \vect T, \alpha)) \geq 1 - \alpha.
\end{equation}
\end{restatable}

\begin{proof}
By construction,
\begin{equation}
\left \{ \theta_{syn} \in \mathcal{C}(\vect R, \vect T, \alpha) \right\} = \left\{ c^-( \vect q( \vect R, \vect T ), \vect J^-_{\theta_{syn}}) \leq N_{S(\vect R)}(\vect T) - \theta_{syn} \leq c^+( \vect q( \vect R, \vect T ), \vect J^+_{\theta_{syn}} )\right\}.
\end{equation}
Therefore, 
\begin{equation}
\Prob ( \theta_{syn} \in \mathcal{C}(\vect R, \vect T, \alpha)) = 
\Prob \left( c^-( \vect q( \vect R, \vect T ), \vect J^-_{\theta_{syn}})  \leq N_{S(\vect R)}(\vect T) - \theta_{\syn} \leq c^+( \vect q( \vect R, \vect T ), \vect J^+_{\theta_{syn}})\right) \geq 1-\alpha.
\end{equation}
The final inequality follows, under the model, from Proposition \ref{prop-critregion}.
\end{proof}

\subsubsection{Sketch of inhibitory confidence interval and computational implementation}

In Section~\ref{sec:mono_causal_infer_model}, we modeled inhibition as a process that censors the elements of $\vect B$ where $\vect B$ is given the same properties as in the excitatory case. This approach is inspired in part by~\citet{spivak2022deconvolution} and, because of the superposition principle of Poisson processes~\citep{daley2008introduction}, similar models are implied by CCG-based methods with heuristic reliance on Poisson assumptions that identify inhibition via short-latency troughs in the CCG~\citep{Bartho2004,tamura2004presumed,kobayashi2019}. However, we should regard this model with much greater skepticism, as inhibitory neurons may play a greater role in regulating downstream spike timing~\citep{wehr2003balanced,pouille2001enforcement} and few experimental studies exist that include causal manipulations of inhibitory neurons~\citep{dudok2021recruitment} in a manner relevant to functional connectivity. Nonetheless, we will sketch an algorithm for computing a bound for inhibition, particularly since it has been hypothesized that axo-axonic cells may function to precisely censor principal cell output \textit{in vivo}~\citep{dudok2021recruitment}. Since the problem has a tenuous empirical foundation, we forgo any rigorous probabilistic interpretation of the algorithm's output. Hence, we prioritize assumptions here that permit the simplest, rather than the most precise, articulation of a concept. Perhaps experimentalists who wish to provide ground truth data for this problem may find this sketch useful while designing their experiments. Varying assumptions and comparing the resulting bound with measurements taken under experimental interventions might help to refine the assumptions needed for an inhibitory model (see Section~\ref{sec:ideal_experiment}).  We emphasize that, for now, this sketch is wholly supported by intuition and simulation, guided by the philosophy that mathematical precision should yield to scientific constraints, which, as just explained, are scarcely available for inhibition.

As in Section~\ref{sec:excitatory_ci}, we must again consider what beliefs about $\vect B$ reasonably explain how the data were generated for inhibition. In particular, here we will suppose some elements of $\vect B$ are censored in empty synchrony regions and we must calculate null distributions corresponding to those suppositions and identify limiting cases to bound $\theta_{syn} < 0$. The problem is not directly analogous to before because, in the excitatory case, the candidate hypotheses for $\vect B$ include all possible subsets of $\vect T$. However, for inhibition, the candidate hypotheses for inhibition include all possible subsets of $S(\vect G)$ where $\vect G \coloneqq \{r : r \in \vect R, N_{S(r)}(\vect T) = 0\}$; that is, all time points in empty synchrony regions in the observed train $\vect T$. It is here that we will make a significant simplification and suppose, by assumption, that censored elements of $\vect B$ are simply a subset of $\vect G$, and compute intervals in simulation thinking of this as an approximation. Inherent in this approximation is we ignore edge effects and assume that no more than one event of $\vect B$ is censored per synchrony region.

Define $\vect Z \coloneqq \vect T \cup \vect G$. $\vect Z$ now represents approximate candidate locations of background events for the excitatory and inhibitory models simultaneously (no generality is lost for the excitatory case in this notation). As earlier let an index set be $\Tilde{\vect K} = \{1,2,...,|\vect Z|\}$ and define a bijective mapping  $f_1: \Tilde{\vect K} \mapsto \vect Z$ as follows. Referring to $f_1(k)$ as $Z_k,$ let $f_1$ be any such mapping that satisfies,
\vspace{-\smallskipamount}
\begin{equation}
\begin{gathered}
   \underbrace{q(\vect R, Z_1) \geq q(\vect R, Z_2) \geq ... \geq q(\vect R, Z_{|\vect G|})}_{Z_k \in \vect G \text{ for } 1 \leq k \leq |\vect G|} \\
    \underbrace{q(\vect R, Z_{|\vect G| + 1}) \leq q(\vect R, Z_{|\vect G| + 2}) \leq ... \leq q(\vect R, Z_{|\vect G| + N_{S(\vect R)}(\vect T)})}_{Z_k \in \vect T \cap S(\vect R) \text{ for } |\vect G| + 1 \leq k \leq |\vect G| + N_{S(\vect R)}(\vect T)}.
\end{gathered}
\label{eq:rankord_inhib}
\end{equation}
\vspace{-\smallskipamount}
\noindent Notice there are no constraints on the mapping for the largest $|\vect T| - N_{S(\vect R)}(\vect T) - 1$ elements of $\Tilde{\vect K}$ which $f_1$ maps to the (non-synchronous) points in  $\vect T \setminus S(\vect R)$. (A  mapping such as $f_1$ always exists since $\vect G$ and $\vect T \cap S(\vect R)$ are mutually exclusive.) With these simplifications, we can once again imagine two limiting cases of how the data might have been generated given a hypothesis that $\theta_{syn} = h$,

\begin{align}
    \Tilde{\vect J}^{-}_h &:=  \begin{cases} 
       \big(\bigcup\limits_{i =  |\vect G| + 1}^{|\vect G| + N_{S(\vect R)}(\vect T) - h} i\big) \cup \big(\bigcup\limits_{i = |\vect G| + N_{S(\vect R)}(\vect T) +  1}^{|\vect G| + |\bm{T}|} i\big), \text{ if } h > 0\\ 
       \bigcup\limits_{i = |\vect G| + h + 1}^{|\vect G|+ |\bm{T}|} i, \text{ if } h \leq 0
       \end{cases} \\
    \Tilde{\vect J}^{+}_h &:= \begin{cases}
        \bigcup\limits_{i=|\vect G|   + h + 1}^{|\vect G| + |\bm{T}|} i, \text{ if } h \geq 0\\ 
        \big( \bigcup\limits_{i= 1}^{-h} i \big) \cup \big(\bigcup\limits_{i = |\vect G| +  1}^{|\vect G| + |\bm{T}|} i\big), \text{ if } h < 0.
        \end{cases}
\end{align}

For a hypothesis of the form $\theta_{syn} = h$ we will posit the existence of some censored background events, with associated probabilities that will then need to be convolved with a function representing a proposal about the distribution of $N_{S(\vect R)}(\vect B)$. As before, these hypotheses are made at limiting cases where the spike counts of $\vect R$ on $\Delta$ timescales are either minimal or maximal (this is builtin to Eqs.~\eqref{eq:rankord1}-\eqref{eq:rankord2}). Noting that we give no rigorous interpretation of these probability statements, let us use the notation $(\ast_{i \in \mathbbm{N}} \vec{v}_i)(k)= \vec{v}_1 * \vec{v}_2 * \vec{v}_3 ...$ to denote the convolution of many vectors where $k$ runs over the support of the resulting vector. In particular, consider the vectors $\vec{v}_i = (1-q(Z_i),q(Z_i))$ for $i \in \Tilde{\vect K}$ and define,
\begin{align}
   \Tilde{c}^{-}(\vect q( \vect r, \vect z ),\vect j) &\coloneqq \max_k \bigg \{ k : \sum_{j=0}^{k} (\ast_{i \in \vect j} \vec{v}_i)(j) \leq \alpha/2 \bigg \} \\
   \Tilde{c}^{+}(\vect q( \vect r, \vect z ),\vect j) &\coloneqq \min_k \bigg \{ k : \sum_{j=k+1}^{|\vect j|+1} (\ast_{i \in \vect j} \vec{v}_i)(j) \leq \alpha/2 \bigg \}. 
\end{align}

\noindent Then let an approximate bound be,

\begin{equation}
\mathcal{\Tilde{C}}(\vect R, \vect Z, \alpha) := \bigg\{ h: \Tilde{c}^{-}( \vect q( \vect R, \vect Z ), \Tilde{\vect J}^-_h) \leq N_{S(\vect R)}(\vect T) - h \leq \Tilde{c}^{+}( \vect q( \vect R, \vect Z ), \Tilde{\vect J}^+_h) \bigg\}.
\end{equation}

\noindent While presented under distinct notation to incorporate inhibition, this set is identical to the rigorous confidence interval derived previously for excitation when $\theta_{syn} \geq 0$. Using this notation, we present two algorithms for efficient computation of these confidence intervals. The algorithms use various tricks to minimize redundant computations and leverage state-of-the-art methods for fast and accurate tail probability computations for a sum of independent random variables. A detailed description of this approach, along with its rationale, is provided in Appendix~\ref{app:algorithm_rationale}. For the reader interested in direct application, we state Algorithm~\ref{alg:convolution} and ~\ref{alg:confidence} immediately without explanation. Algorithm~\ref{alg:confidence} is the main algorithm that computes, as an example, the lower bound for $\theta_{syn}$ in the excitatory case. For clarity, Algorithm~\ref{alg:convolution} is separated but repeatedly called from Algorithm~\ref{alg:confidence} and houses machinery for accurately computing tail areas for sums of independent random variables.

\begin{algorithm}
\caption{Hybrid convolution power with shift plus sparse exceptions}\label{alg:convolution}
\begin{algorithmic}[1]
\Statex \textbf{Input:} Scalar $L_{div}$ for $L_{div}$-fold convolution power, $\vec{p} = (p_0,p_1,...,p_N)$ probability vector to apply $L_{div}$-fold convolution power, $\vec{g} = (g_0,g_1,...g_M)$ residual probability vector, $y_0$ the observed test statistic
\Function{Sdpnt}{$s,\vec{p}, \vec{g}$}
\State $\kappa_{\vec{p}}^{'}(s),\kappa_{\vec{g}}^{'}(s) \leftarrow$ the derivatives of the CGFs of $\vec{p}$ and $\vec{g}$ evaluated at $s$
\State \Comment{Note: CGF stands for cumulant generating function}
\State \Return $L_{div}\kappa_{\vec{p}}^{'}(s) + \kappa_{\vec{g}}^{'}(s) - y_0$
\EndFunction
\State $\hat{s} \leftarrow$ compute $s$ such that \textit{Sdpnt}$(s,\vec{p}, \vec{g}) = 0$
\ForAll{$x \in \{0, \dots, N\}$} \Comment{apply exponential tilts}
    \State $\vec{p}_{\hat{s}}(x) \leftarrow \exp{[\hat{s}x - L_{div}\kappa_{\vec{p}}^{'}(\hat{s})+ \kappa_{\vec{g}}^{'}(\hat{s})]}\vec{p}(x)$
\EndFor
\ForAll{$x \in \{0, \dots, M\}$} 
    \State $\vec{g}_{\hat{s}}(x) \leftarrow \exp{[\hat{s}x - L_{div}\kappa_{\vec{p}}^{'}(\hat{s})+ \kappa_{\vec{g}}^{'}(\hat{s})]}\vec{g}(x)$
\EndFor
\State $N_{\vec{b}} \gets 2^{\lceil \log_{2}(K(N-1)+(M-1)+1) \rceil}$
\State $\vec{p}_{\hat{s}} \leftarrow $\textit{Concat}$(\vec{p}_{\hat{s}},\vec{0}_{L_{\vec{p}}})$ pad with zero vector of dimension $L_{\vec{p}} = N_{\vec{b}} - (N+1)$ 
\State $\vec{g}_{\hat{s}} \leftarrow $\textit{Concat}$(\vec{g}_{\hat{s}},\vec{0}_{L_{\vec{g}}})$ where $L_{\vec{g}} = N_{\vec{b}} - (M+1)$ 
\State $\vec{b}_{\hat{s}} \leftarrow D^{-1}((D\vec{g}_{\hat{s}}) \odot (D\vec{p}_{\hat{s}})^{\odot L_{div}})$ \Comment{computed via FFT, IFFT}
\State \Comment{$D$, $D^{-1}$, \& $\odot$ are the DFT, IDFT, and pointwise product respectively}
\ForAll{$x \in \{0, \dots, N_{\vec{b}}\}$} \Comment{reverse tilt}
    \State $\vec{b}(x) \leftarrow \exp{[\kappa_{\vec{p}}^{'}(\hat{s})+ \kappa_{\vec{g}}^{'}(\hat{s}) - \hat{s}x]} \vec{b}_{\hat{s}}(x)$
\EndFor
\State pval $\gets \sum_{i=y_0}^{N_{\vec{b}}} b_i$ 
\Statex \textbf{Output:}  pval
\end{algorithmic}
\end{algorithm}

\begin{algorithm}
\caption{Coarse-to-fine lower confidence bound computation}\label{alg:confidence}
\begin{algorithmic}[1]
\Statex \textbf{Input:} Spike trains $\vect R$ and $\vect T$, $\alpha$ (1-desired confidence level), $z_0$ (the measured value of $N_{S(\vect R)}(\vect T)$), $L_{div}$ (hyperparameter)
\State $(q(Z_1),q(Z_2),...) \leftarrow$ compute from $\vect R$ \& $\vect T$ for an arbitrary $\vect G \in \mathcal{G}$ 
 \State Implement \texttt{BinarySearch} for the case  $\Tilde{\vect J}_h^{+}$ on $h \in (1,2,...,z_0)$ with search query: 
 \Statex $h^*=\min_h \big\{ h : \big[f(h) + \mathbbm{1}\{\sum_{i \in \Tilde{\vect J}_h^{+}} q(Z_i)>y\}\left(1-2f(h)\right)\big]> \alpha/2 \big\}$
  \Statex where $f(h) = \exp \big[-\sum_{i \in \Tilde{\vect J}_h^{+}} q(Z_i) + y + y\ln(\frac{1}{y} \sum_{i \in \Tilde{\vect J}_h^{+}} q(Z_i))\big] \text{ \& } y = z_0 - h$
 \State $CF_L \leftarrow$ $h^*$ the result of \texttt{BinarySearch}
 \State $CF_U \leftarrow$ $h^{**}$ the result of analogous \texttt{BinarySearch} for the case $\Tilde{\vect J}_h^{-}$
 \State \(\vec{u} = (u_1,u_2,...) \gets \{ q(Z_i) : i \in \Tilde{\vect J}^{+}_{CF_{U}} \} \)\Comment{Unique probabilities for $i \in \Tilde{\vect J}^{+}_{CF_{U}}$}
\ForAll{$i \in \{1,2,...,\dim(\vec{u})\}$}
    \State $n_i \gets \sum_{j \in \Tilde{\vect J}^{+}_{CF_{U}}} \mathbbm{1}\{q(Z_j) = u_i\}$
    \State $m_i \gets \left\lfloor n_i/L_{div} \right\rfloor$
    \State $w_i \gets n_i - m_i L_{div}$
\EndFor

\State $\vec{p} \gets$ apply DC to binomials with $n = m_i$ \& $p = u_i$ for $i \in \{1,2,...,\dim(\vec{u})\}$  
\State $\vec{a} \gets$ apply DC to binomials with $n = w_i$ \& $p = u_i$ for $i \in \{1,2,...,\dim(\vec{u})\}$  
\State Implement \texttt{BinarySearch} on $h \in (CF_L,CF_L+1,...,CF_U)$ with search query: 
 \Statex $h^{***} = \min_h \big\{ h : \mathbbm{P}(N_{S(\vect R)}(\vect T)-h \geq z_0-h |(q(X_1),q(X_2),...), \Tilde{\vect J} = \Tilde{\vect J}_h^{+}) > \alpha/2 \big\}$
 \Statex where given each $h$ tested the tail probability is computed via sub-steps:
 \Statex \;\;\;\fontsize{9}{11}\selectfont(13.1) \( y_0 \gets z_0 - h \)
\Statex \;\;\; \fontsize{9}{11}\selectfont (13.2)  \( \vec{g} \gets \) apply DC to \( \vec{a} \) and the distributions w/ success prob. \( q(Z_i) \) for \( i \in \Tilde{\vect J}^{+}_{h} \setminus \Tilde{\vect J}^{+}_{CF_{U}} \)
\Statex \;\;\; \fontsize{9}{11}\selectfont (13.3) tail probability \( \gets \) Pass \( \vec{p} \), \( \vec{g} \), \( y_0 \), and \( L_{div} \) to Algorithm 1 \normalsize

\State Lower confidence bound $\leftarrow$ $h^{***}$ the result of \texttt{BinarySearch}
\Statex \textbf{Output:} Lower confidence bound
\end{algorithmic}
\end{algorithm}

\newpage
\section{Causal inferences in simple simulated systems}

In the following sections, we will study the monosynaptic causal inference model in simulation. The non-parametric nature of the model possesses some features that may seem foreign to those accustomed to modeling point processes with objects such as conditional intensity functions and generalized linear models (GLMs). For example, to define a background timescale, we partitioned time into arbitrarily phased intervals, each of duration $\Delta$. The model makes no use of conditional intensity functions and few assumptions were made about the interaction process $\vect I^{(\vect R)}$. At first, we will demonstrate the features just mentioned are appropriate in a conditional intensity model ensuring the process follows the monosynaptic causal inference model's assumptions only at the level of analogy. Later, we will test inferences in neural dynamical systems where it is much less clear if the assumptions are appropriate, and thus, various validations will be necessary.

\subsection{Causal inferences in a conditional intensity function model}
\label{sec:point_process_sim1}

The conditional intensity model here will exhibit rapid nonstationarities with random phases, but the nonstationary fluctuations will have timescales with $\Delta$ as a lower bound. This will suggest the idealized construction of $\gamma(t)$ with fixed $\Delta$ is scientifically appropriate. One can also find theoretical arguments supporting this construction in past work~\citep{Platkiewicz2017}. The conditional intensity functions are made smooth by generating them from normalized Ornstein-Uhlenbeck processes with non-stationary means that define the background timescales. The synaptic coupling is generated by convolving the presynaptic spike trains with a truncated exponential kernel. The notion of a synapse with finite strength at infinite decay time will be abstracted out now and will naturally reenter the study of dynamical systems models later. In addition to confounding background excitability fluctuations, here, the synaptic kernel has a private background excitability function, modeling a situation where the postsynaptic dendritic compartment may have its own excitability fluctuations, confounding the causal relationship. These are all ways to confound the relationship between spike trains while maintaining a type of separation of timescale.

Motivated by the observation that populations of neurons have downstates and upstates ~\citep{luczak2007sequential,hromadka2013up}, let us introduce nonstationary fluctuations between neurons on coarse timescales with varying degrees of skew.  Furthermore, consider that monosynaptically-interacting neurons might be at different phases of an oscillation in the local field potential (e.g., hippocampal gamma oscillations~\citep{harris2003organization}). This example is concrete and empirical, but it is easy to imagine complex neural computations generate other types of confounding in pairwise interactions. Coarse timescale nonstationarities with complex dependence structure can then generate confounding in the CCG even when a separation of timescales assumption is true.  

To generate some limiting cases in simulation, consider a sequence of multivariate skew random variables,
\begin{align}
\label{eq:multivaraiateskew}
\vec{\vect  M}_{n,k} = \begin{bmatrix}
m_{0,k}  \\
m_{1,k} \\
\vdots  \\
m_{n-1,k}
\end{bmatrix} \sim  2 \phi_n(\vec{m}_{n,k};\vec{\vect \Omega}) \Phi(\vec{\alpha}^\top \vec{m}_{n,k}), \vec{m}_{n,k} \in \mathbbm{R}^n, k \in \mathbbm{Z}^*
\end{align}

\noindent where $\phi_n(\vec{m}_{n,k};\vec{\vect \Omega})$ is a zero mean $n$-dimensional normal density with correlation matrix $\vec{\vect \Omega}$, $\Phi(\cdot)$ is the standard normal distribution function, and $\vec{\alpha}$ is an $n$-dimensional vector that controls skewness~\citep{azzalini1999statistical}. Let $L_i \sim U(\Delta,a\Delta)$ ms, $a>1$ and partition $\mathbbm{R}^+$ into contiguous disjoint intervals $\{[x_{k-1},x_k)\}_{k=0}^{\infty}$ such that $x_k = x_{k-1} + L_k$ and $x_0 = 0$. Associate with each $L_k$ a sample from the multivariate skew distribution, $\vec{\vect  M}_{i,k}$, with dimension $n=3$ and define a set of background excitability functions as,
\begin{equation}
    b_i(t) = \sum_{k \in \mathbbm{Z}^*} m_{i,k}\mathbbm{1}\{t \in [x_{k-1},x_k)\}, \text{ for } i = 0,1,2.
\end{equation}

\indent In simulation, we thus imagine that the excitability of the reference neuron, $b_0(t)$, target neuron, $b_1(t)$, and dendritic compartment of a synapse between them, $b_2(t)$, might have arbitrary confounding fluctuations on coarse timescales generated from a multivariate skew. That is, $b_0(t), b_1(t), b_2(t)$ might be skewed - i.e., rare up or down states~\citep{DeWeese2006} - and these states may have positive or negative correlations with each other.

Define $\tau_d$ as a conduction delay, $\tau_s$ as a phenomenological synaptic relaxation time, and $\upsilon(\tau) = \exp(-\tau/\tau_s) \mathbbm{1}\{0 \leq \tau < \tau_{mx}\}$ as a synaptic kernel zero everywhere but $\tau \in [0,\tau_{mx})$. The model is then defined by the conditional intensity functions,
\begin{align}
    \lambda_R(t) | U_0(t) &= 
    \rho_0 U_0(t) \label{eq:toy_rate_model2_begin} \\
    \lambda_{T}(t) | U_1(t), U_2(t), do(\vect R=\vect r) &= \rho_1 U_1(t_j) + \epsilon U_2(t) \int_{-\infty}^{\infty}\upsilon(t-\tau-\tau_d) \sum_{r \in \vect r} \delta_d(t-r) d\tau 
\label{eq:toy_rate_model2_end}
\end{align}

\noindent where $\rho_0, \rho_1$ are normalization factors, $\epsilon$ is a coupling constant,

\begin{equation}
    \tau_I \frac{dI_i(t)}{dt} = -I_i(t) + b_i(t) +\sigma_I \sqrt{2\tau_I} \xi_i(t), \xi_i(t) \sim \mathcal{N}(0,1) 
\end{equation}

\noindent is used to obtain  $U_i(t) = a_iI_i(t) + b_i$, $a_i$ and $d_i$ are chosen to min-max normalize $U_i(t)$, and $\sigma_I^2$ and $\tau_I$ are the variance and timescale of the smoothing agent, respectively. For clarity, we restate the causal interpretation of the model in Eqs. \eqref{eq:toy_rate_model2_begin}-\eqref{eq:toy_rate_model2_end} in relation to the general model stated in Section~\ref{sec:general_causal_model} and Example~\ref{rem:measure_spike_model}. The system is governed by a common probability space and under the subcausal model induced by $do(\vect R = \emptyset)$ Eq.~\eqref{eq:toy_rate_model2_end} reduces to its first term.

The model is simulated in Figure \ref{fig:point_process_D}. Figure~{\ref{fig:point_process_D}A} depicts the coupled conditional intensity model in a cartoon fashion. Each point in  Figure~{\ref{fig:point_process_D}B} represents a distinct simulation as follows. For each simulation, the Vine Beta method, with its parameter fixed to $0.1$, is used to generate a random covariance matrix with strong correlations~\citep{lewandowski2009generating} scaled to set $\vec{\vect \Omega}$ and we sample $\vec{\alpha} \sim U([0,100(2a -1))^n])$ where $a \sim Be(0.5)$. We then sample the sequence of independent and identically distributed (within but not across simulations) multivariate skew vectors $\vec{\vect M}_{n,k}$ and construct the background excitability functions. As mentioned previously, the motivation for the particular form and parameters of the simulation is to induce confounding background fluctuations between the simulated neurons, with skewed up and down states, and confounded state-dependent synaptic efficacy as well. Normalization factors $\rho_0$ and $\rho_1$ are all sampled so that the average firing rates of all spike trains in the absence of coupling are uniformly distributed between $50$ and $200$ spikes/second. Finally, spike are simulated from $\lambda_R(t) | U_0(t), \lambda_T(t) | U_1(t), U_2(t), \vect R$, and $\lambda_T(t) | U_1(t), U_2(t), do(\vect R = \emptyset)$. From these we obtain the spike trains $\vect R, \vect T$, and $\vect T^{(\emptyset)}$. We then compute the simulated ground truth $\theta_{syn}$, along with $\hat{\theta}_{syn}$. Here we assume $\tau_{mx}$ is known, and so the statistical free parameter $\delta$ was made one time bin larger. We also assume knowledge of $\Delta$ and hence set the statistical parameter equal to the value used to simulate the process (which is a lower bound). For each simulation a point estimate is then computed with Eq.~\eqref{eq:pointestimator} and 95\% confidence intervals are computed with Algorithms~\ref{alg:convolution} \& \ref{alg:confidence}. For 101 simulations the empirical coverage probability of the confidence intervals is 0.98.

\begin{figure}[!htbp]
    \centering
\includegraphics[width=1\textwidth]{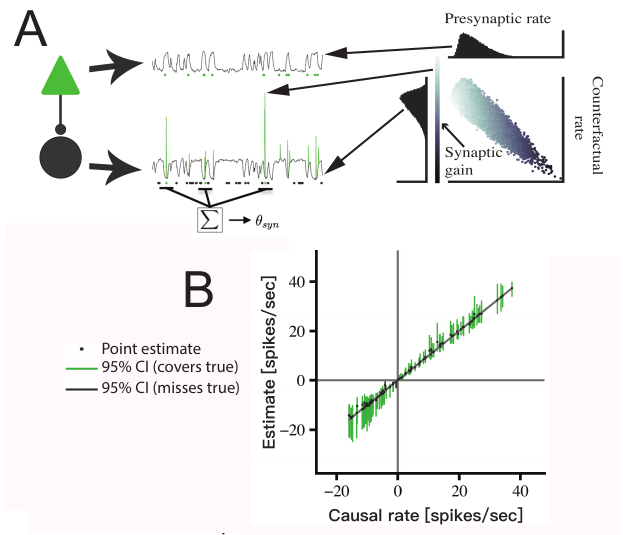}
    \caption[Point process demonstration]{{\bf Point process demonstration.} \textbf{A:} A cartoon depiction of the conditional intensity model described in Eqs. \eqref{eq:toy_rate_model2_begin} - \eqref{eq:toy_rate_model2_end}. $
    \lambda_R(t)$  and $
    \lambda_T(t)|do(\vect R = \emptyset)$ are denoted in black. The synaptic excitability function $b_2(t)$ is not depicted but modulates the height of the synaptic gain (the green gain on top of $
    \lambda_T(t)|do(\vect R = \emptyset)$). To generate confounding, the coarse timescale amplitudes of background rates and synaptic efficacy are generated from a multivariate skew distribution with strong correlations. The synaptic gain modulation is represented as the gray colormap. \textbf{B:} One hundred and one simulations from the model of Eqs. \eqref{eq:toy_rate_model2_begin} - \eqref{eq:toy_rate_model2_end} with empirical coverage probability $0.98$. Point estimates are always shown, whereas confidence intervals are not drawn if the null hypothesis $H_0: \theta_{syn} = 0$ fails to reject.}
    \label{fig:point_process_D}
\end{figure}

\subsection{Mapping the statistical model onto a dynamical system}
\label{sec:map_to_lif}
Thus far, it has been assumed that a postsynaptic spike train is derived from a latent mixture of background events, $\vect B$, and interactions, $\vect I^{(\vect R)}$. It should be clear from the assumptions and from the demonstration in a conditional intensity model that the division into two classes does not at all mean a constant synaptic weight; the model houses conditional intensity models where events analogous to $\vect I^{(\vect R)}$ might arise from different state-dependent probabilities. Rather, the idealization lies in positing that there exist some events $\vect B$ that may be assumed to have zero effective causal weight, and both classes have timescale assumptions that make $\theta_{syn}$ identifiable. This might be described as a type of causal coarsening. However, despite its clear merits in terms of analytic tractability, the clean division of the postsynaptic train into two classes gave rise to three free parameters $\delta$, $\tau$, and $\Delta$. By assumption, all interaction points are confined to be members of the set $S(\vect R,\delta,\tau)$ whereas $\Delta$ defines the background timescale. While constructed from conditional intensity functions, the simulated model of the previous section more or less ensured by construction that causal spikes would be confined to a set $S(\vect R,\delta,\tau)$  and non-causal spikes would possess no temporal structure for timescales smaller than $\Delta$ where $\delta < \Delta$.

It is now natural to challenge aspects of that idealization in some settings even more foreign to the one in which the model was derived. A sensible choice is dynamical neuron models, which well-capture features of cortical neurons~\citep{Gerstner2009} and where ground truth causal information is available by recycling the concept of frozen noise to be applied to stochastic input currents. We first must ask if a $\tau$ and $\delta$ can be chosen such that the simple division of a postsynaptic train into events $\vect B$ and $\vect I^{(\vect R)}$ might approximate the causal action of a presynaptic input through a dynamical system. The second question to consider is if, given knowledge of $\delta$ and $\tau$, a $\Delta$ can be chosen to recover causal counterfactual quantities in the midst of confounding. 

We do not provide a method to choose $\delta,\tau$, and $\Delta$ from first principles with observational data and are skeptical that the task is even possible, particularly in the case of $\Delta$. Rather, we regard them as free parameters in the physicist's sense. So the task here is bent toward understanding the qualitative mapping of these free parameters onto some mechanistic features. As such, in contrast to our previous work~\citep{platkiewicz2021monosynaptic}, as a matter of interpretation and robustness, we study commonly used dynamical mechanisms (e.g., LIF, EIF, AdEx) throughout future sections. This will demonstrate that the statistical framework's validity is also a matter of qualitative considerations. For example, when one asserts that the synaptic process  $\vect I^{(\vect R)}$ is fast, the word \textit{fast} clearly has meaning only relative to the background timescale and hence what is of true theoretical interest is the ratio of the effective synaptic and background input timescales. The reader should make judgments about the quantitative plausibility from 
real data, for example, by consideration of the fine-timescale effects studied by \citet{English2017}, for which standard integrate-and-fire type models might be insufficient~\citep{platkiewicz2021monosynaptic}.

\subsubsection{System of feedforward leaky integrate-and-fire (LIF) neurons}

Consider first a system of standard leaky integrate-and-fire (LIF) neurons with instantaneous conduction delay driven by background input currents and a synaptic conductance from a presynaptic neuron ($i=0$) into a postsynaptic neuron ($i=1$),
\begin{align}
    C_m \frac{dV_i}{dt} &= -g_l(V_i-E_l) -g_s g_0 (V_i-E_{syn}) \mathbbm{1}\{i=1\} + I_i, \text{ for } i = 0,1  \label{eq:lif_system_begin}
 \\
\tau_{syn} \frac{dg_s}{dt} &= -g_s + \sum_{r \in \mathbf{R}}(1-g_s)\delta_d(t - r)
\label{eq:lif_system_end}
\end{align}

\noindent where $V_i$ is the voltage (mV) of neuron $i$, $g_l$ and $g_s$ are the leak and synaptic conductances (mS/cm$^2$), $E_l$ and $E_{syn}$ are the leak and synaptic equilibrium potentials (mV), $C_m$ is the specific capacitance ($\mu$F/cm$^2$), $g_0$ and $\tau_{syn}$ are the peak synaptic conductance (mS/cm$^2$) and timescale (ms), and $I_i$ is the background input ($\mu$A/cm$^2$) to neuron $i$.\footnote{We continue to work in these units throughout.} If at time $t_0$, $ V(t_0) = V_{T}$, a spike is tabulated followed by the reset condition $\lim_{\{\epsilon \rightarrow 0 : \epsilon > 0\}} V(t_0 + \epsilon) = E_l$ and clamped there for a refractory period of $\tau_r$. We map the system Eqs.~\eqref{eq:lif_system_begin}-\eqref{eq:lif_system_end} onto the monosynaptic causal model by identifying
$\vect R = \{t : V_0(t) = V_{T} \}$, $\vect T = \{t : V_1(t,\vect R) = V_{T} \}$, and $\vect T^{(\emptyset)} = \{t : V_1(t,do(\vect R =\emptyset)) = V_{T}\}$ where the last line refers to the trajectory $V_1(t)$ under the deterministic modification of Eq.~\eqref{eq:lif_system_end} $do(\vect R = \emptyset) \implies g_s=0, \forall t \in \mathbbm{R}$ in the sense that the voltage trajectories may change but the $I_i$ remain constant (i.e., frozen) for all realizations of the stochastic input current. In other words, unlike in Remark~\ref{rem:measure_spike_model}, here, the probability space is defined directly over the background input current, and the dynamical system determines the functional relationship between the reference and target train. 
To study the interpretation of the monosynaptic causal inference model in terms of dynamical mechanisms, here we simplify the form of the background model while maintaining confounding common input,
\begin{align}
    I_i &= U_i + U_2, \text{ for } i \in \{0,1\} \\
    \tau_{I,i} \frac{d U_i}{dt} &= U_i + \mu_i + \sigma_i \sqrt{2\tau_{I,i}} \xi_i(t), \text{ for } i \in \{0,1,2\}
    \label{eq:ou_process}
\end{align}

\noindent where $\xi_i(t) \sim \mathcal{N}(0,1)$ and $U_0,U_1,U_2$ are all independent Ornstein-Uhlenbeck processes with means $\mu_i$, variances $\sigma_i^2$, and timescales $\tau_{I,i}$. It is not entirely clear from Eq.~\ref{eq:lif_system_end} where one ought to look for spiking events that may be well-attributed to a synaptic process like $\vect I^{(\vect R)}$. Certainly, we may first assume that they tend to occur after and not before the spikes $\vect R$. Furthermore, we would imagine that their tendency to occur decreases as a function of distance from the spikes $\vect R$ given the exponential decay model of synaptic conductances. These are elementary and routinely applied assumptions but do not yet appeal to causal inference concepts. Let us make the simplification of instantaneous biophysical conduction delay, $\tau_d = 0$, and measure counterfactual spike counts in the spectrum of sets, $\{S(\vect R,2 \delta,\delta) : \delta \in [0,\infty)\}$. For this purpose define the function,
\begin{equation}
    g(\delta) = N_{S(\vect R,2\delta,\delta)} (\vect T) - N_{S(\vect R,2\delta,\delta)} (\vect T^{(\emptyset)}).
\end{equation}

\noindent Assume for now $g(\delta)$ is monotone and that
\begin{equation}
    \Bar{\beta} = \lim_{\delta\rightarrow \infty} |g(\delta)|
\end{equation}

\noindent exists. Then for any $0 \leq \beta_0 < 1$ let,

\begin{equation}
    \label{eq:beta_hyper}
         \delta_0 = \begin{cases}
             \min_{\delta} \big\{\dfrac{\delta}{2} : \dfrac{|g(\delta)|}{\Bar{\beta}} = \beta_0\big\} &\text{if } \exists \delta \in \mathbbm{R}^+ \text{ s.t. } \dfrac{|g(\delta)|}{\bar{\beta}} = \beta_0\phantom{aaa}\\[2ex]
             0 & \text{otherwise}.
        \end{cases}
\end{equation}

\begin{figure}
    \centering
\includegraphics[width=.7\textwidth]{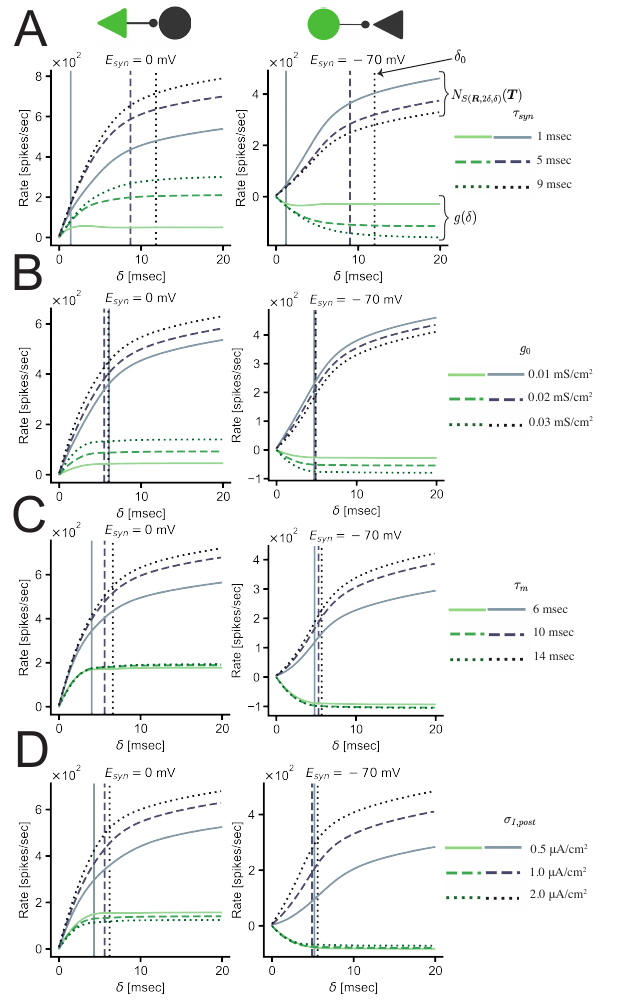}
    \caption[Mapping the statistical monosynaptic causal inference model onto a leaky integrate-and-fire (LIF) system]{{\bf Mapping the statistical monosynaptic causal inference model onto a leaky integrate-and-fire (LIF) system.} Functions $g(\delta)$ and $N_{S(\vect R,2\delta,\delta)}(\vect T)$ normalized by a factor $|\vect R| \cdot dt$ for $\delta \in [0, 20]$ msec are plotted in green lines and black lines respectively. Vertical lines are $\delta_0$ for $\beta_0 = 0.95$. Different line styles correspond to three variations of one postsynaptic parameter per plot. Each biophysical parameter setting corresponds to a long simulation of $27.77$ simulated hours. The left column is for an excitatory synapse ($E_{syn} = 0$ mV), and the right is for an inhibitory synapse ($E_{syn} = -70$ mV). \textbf{A:} The synaptic decay time constant $\tau_{syn} \in \{1,5,9\}$ msec. \textbf{B:} The peak synaptic conductance $g_0 \in \{0.01,0.02,0.03\}$ mS/cm$^2$. \textbf{C:} The membrane time constant $\tau_m \in \{6,10,14\}$ msec. \textbf{D:} The postsynaptic noise amplitude $\sigma_{I,post} \in \{0.5,1,2\}$ $\mu$A/cm$^2$ (referred to as $\sigma_{I,post}$ in the figure but defined as $\sigma_{1}$ in the text).}
    \label{fig:dynamical_mapping}
\end{figure}

Intuitively, $\delta_0$ describes how large the time interval after the reference spikes must be to capture some proportion, $\beta_0$, of the causal difference in spike counts between the counterfactuals relative to the causal difference at some long-term value $\Bar{\beta}$.

\subsubsection{Simulations validating the mapping}

Figure~\ref{fig:dynamical_mapping} shows simulations of Eq.~\eqref{eq:lif_system_begin}-\eqref{eq:lif_system_end} and plots $N_{S(\vect R,2\delta,\delta)} (\vect T)$ and $g(\delta)$ for $\delta \in [0,20]$ ms. Both these functions are normalized by a factor $|\vect R| \cdot dt$. Each panel displays six lines:  $N_{S(\vect R,2\delta,\delta)} (\vect T)$ and $g(\delta)$ for three different values of a dynamical parameter. Vertical lines mark $\delta_0$, which is obtained by setting $\beta_0 = 0.95$ and where $\Bar{\beta}$ is approximated by taking $|g(\delta)|$ at $\delta = 1000$ ms. Plots are shown for both excitatory ($E_{syn} = 0$ mV) and inhibitory ($E_{syn} = -70$ mV) synapses. In the LIF neuron, $g(\delta)$ tends to rise roughly monotonically and saturate quite quickly, although different dynamical parameters have a significant impact, particularly the synaptic timescale $\tau_{syn}$.

Figure~{\ref{fig:dynamical_mapping}A} demonstrates simulations of the LIF model system with $\tau_{syn} \in \{1,5,9\}$ ms. For both the excitatory and inhibitory synapses, at $\tau_{syn} = 19$, $\delta_0 \approx 2$ ms whereas at $\tau_{syn} = 9$, $\delta_0 \approx 12$ ms. That long synaptic decay times increase the timescale of causal action on the postsynaptic neuron coincides with intuition. Figure~{\ref{fig:dynamical_mapping}B} shows the effect of peak synaptic conductance, $g_0$, on $\delta_0$; the effect on $\delta_0$ is less dramatic than $\tau_{syn}$. For $g_{0} \in \{0.01,0.02,0.03\}$ mS/cm$^2$, $\delta_0$ clusters around $\delta_0 \approx 7$ ms for an excitatory synapse and $\delta_0 \approx 6$ ms for an inhibitory synapse. Figure~{\ref{fig:dynamical_mapping}C} shows the analogous plots for membrane time constant, $\tau_m$. The effect on $\delta_0$ is slightly more pronounced than peak synaptic conductance but still less than synaptic decay time. In the excitatory case, for $\tau_{m} = 6$ ms, $\delta_0 \approx 4$ ms whereas for $\tau_{m} = 14$ we observe $\delta_0 \approx 6.6$ ms. In the inhibitory case $\delta_0$ is clustered around $\delta_0 \approx 6$ ms with a slight positive trend with $\tau_m$. Figure~{\ref{fig:dynamical_mapping}D} again shows the analogous plots for the postsynaptic Gaussian noise amplitude which is known to influence postsynaptic response dynamics~\citep{Herrmann2002}. A negative trend is seen between the postsynaptic noise $\sigma_{1} \in \{0.5,1,2\}$ $\mu$A/cm$^2$ (refered to as $\sigma_{I,post}$ in the figure) and $\delta_0$. No trend is detected in the inhibitory case.

Figure~{\ref{fig:lif_estimates}} plots monosynaptic point estimates and confidence intervals for the LIF system varying the same parameters as the previous plot. To make sensible comparisons across plots, here parameters are normalized as $g(\delta_0)/(|\vect R| \cdot dt)$, termed \textit{causal rate} in the figure. Estimates are $\hat{\theta}_{syn}'/(|\vect R| \cdot dt)$ For each plot, different levels of causal rate are produced by varying $g_{0} \in [0,0.1]$ mS/cm$^2$ with eleven equally spaced values for both excitatory ($E_s = 0$ mV) and inhibitory ($E_s = -70$) synapses. For estimation, we must also choose a value for the statistical parameter $\Delta$. Recall that $\mathcal{A}.$\ref{as:2} (timescale separation) requires that some $\tau$, $\delta$, and $\Delta$ exist such that  $\vect I^{(\vect r)} \subset S(\vect r)$, for all $\vect r$. In simulation $\delta_0$ is chosen such that $\beta_0 = 0.95$ to approximate this assumption. On the other hand $\mathcal{A}.$\ref{as:3} (positivity) requires  $0 \leq q(\vect R, k\Delta) < 1$, for all $k \in \mathbbm{Z}^*$. Once $\delta$ is chosen to be $\delta_0$, an observer can choose some $\Delta_0$ and at least verify $0 \leq q(k\Delta_0) < 1$, for all $k \in \mathbbm{Z}^*$ because this condition only requires access to the observed data to verify. This is not as useful as it may seem, however, because, of course, $\Delta$ is unknown, and an appropriate selection of it determines the validity of $\mathcal{A}.$\ref{as:4} (conditional uniformity) which cannot be assessed from observational data. Furthermore, $\mathcal{A}.$\ref{as:2} and $\mathcal{A}.$\ref{as:3} are orthogonal assumptions: one can be true while the other is false. In these simulations, the common background input timescale was chosen to be rather large ($\tau_{I,2} = 50$ ms) so that a simple heuristic might automate the choice of $\Delta_0$ given $\delta_0$ which strongly biases this inquiry toward assessing c. In this figure, without too much thought, we choose $\Delta_0 = \delta_0 + 4$ ms, which typically well-approximates the $\mathcal{A}.$\ref{as:3} (positivity) in the regimes explored here, although some misestimation arises from violations. To study causal identifiability, in all plots to follow, we use $\hat{\theta}_{syn}'$ from Corollary~\ref{thm:my_coro1} as an estimate and always define $\theta_{syn}$ as the ground truth value.

It is worth reminding the reader that $\Delta$ is unknown, and perhaps unknowable, in these simulations. The timescales of the membrane, background input, and synapse likely interact and might even produce a statistical background timescale that is smaller than the timescales of the physiological variables involved. Similarly, the assumption that the inequality $\delta < \Delta$ can be true while satisfying the other assumptions cannot be known in the simulation and, in fact, is one of the primary motivations for testing estimation in dynamical systems models while varying physiological parameters. That is, good estimation is regarded as evidence for the fulfillment of the assumptions.

The qualitative results of Figure~{\ref{fig:lif_estimates}} can, for the most part, be predicted by the results of the previous figure. That is, the estimation procedure provides highly accurate estimates to the degree that the model's assumptions are approximated in the sense of the mapping proposed earlier.  Figure~{\ref{fig:lif_estimates}A} shows point and interval estimates for $\tau_{syn} \in \{1,5,9\}$ ms. As Figure~{\ref{fig:dynamical_mapping}A} predicts, as $\tau_{syn}$ increases $\delta_0$ also increases which puts stress both on $\mathcal{A}.$\ref{as:4} (conditional uniformity) and $\mathcal{A}.$\ref{as:3} (positivity). One ought to heed the point just made about $\Delta$ being unknown in dynamical simulations.

Even at an unrealistically small synaptic timescale of $\tau_{syn} = 1$ ms, the magnitude of the inhibitory causal rate is slightly underestimated with empirical coverage probability $0.82$ for the confidence intervals. As $\tau_{syn}$ increases, both the magnitude of excitatory and inhibitory causal rates are underestimated, however, the confidence intervals behave more conservatively in this regime and have empirical coverage probability $1$ for $\tau_{syn} \in \{5,9\}$ across all simulations. 

In {\ref{fig:lif_estimates}B} we observe that the qualitative behavior of estimation with respect to synaptic timescale $\tau_{syn}$ is recapitulated for membrane timescale $\tau_{m}$ although to a less pronounced degree. That is, Figure~{\ref{fig:dynamical_mapping}C} indicates $\delta_0$ will increase as membrane timescale $\tau_{m}$ increases and accordingly in {\ref{fig:lif_estimates}B}  the magnitude of causal rate is slightly underestimated for excitatory and inhibitory interactions as $\tau_m$ increases.

A notable feature of Figure~{\ref{fig:dynamical_mapping}} was that all else being equal, increasing $\tau_{syn}$ or $\tau_{m}$ increased the causal rate for all $\delta$ and for the most part trended positively with $\delta_0$. More intuitively, as the temporal scale of causal effect (i.e., $\delta_0$) increased more spikes were causal, naturally. But this appears not to be the case for  $\sigma_{1}$ (termed $\sigma_{I,post}$ in the figure). For excitatory interactions, in Figure~{\ref{fig:dynamical_mapping}D} increasing $\sigma_{1}$ increased $\delta_0$, however, the magnitude of causal rate decreased as $\delta_0$ increased unlike in the case of synaptic timescale $\tau_{syn}$ and membrane timescale $\tau_{m}$. Yet, like the case with $\tau_{syn}$ and $\tau_{m}$,  Figure~{\ref{fig:dynamical_mapping}D} and Figure~{\ref{fig:lif_estimates}C} in combination show that estimation is accurate to the degree $\delta_0$ is made small by a small $\sigma_1$. This all suggests, not at all in conflict with intuition, that the model's validity has some mechanistic independence so long as the causal behavior at the level of spiking abides by the formal assumptions proposed earlier.

The idea that causal effects of inhibitory synapses can be estimated from spike trains has far less existing support from \textit{in vivo} experiments. Furthermore, here, the estimation of inhibitory synapses is only highly accurate for unrealistic parameters for inhibitory synapses~\citep{otis1992modulation}. For these reasons, the latter figures primarily focus on the study of excitatory interactions, except in a few idealized settings where the math quickly implies an inhibitory solution. In that case, this cautionary statement still applies. However, it should be noted that physiological parameters interact, and they are typically measured in settings where the interactions may not be present, such as \textit{in vitro} studies. Thus one cannot rule out the possibility that, at the level of spiking, the temporal scale of causal action for inhibitory synapses is still small \textit{in vivo}, meaning the deficiency resides in the biophysical models. However, until more basic evidence of such mechanisms exists, we must remain skeptical that the inhibitory model proposed in this study has any relevance to neuroscience.

\begin{figure}
    \centering
\includegraphics[width=.7\textwidth]{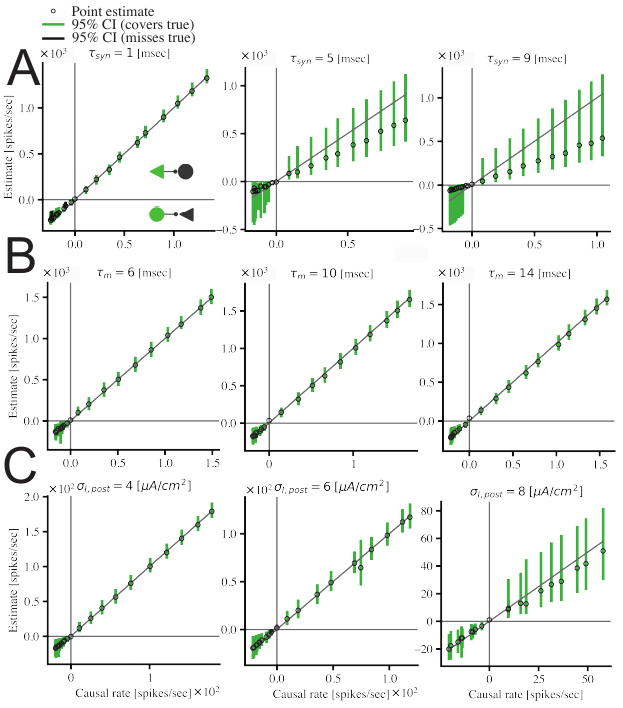}
    \caption[Causal inferences in the LIF model as a function of various dynamical parameters]{{\bf Causal inferences in the LIF model as a function of various dynamical parameters.} Inference of postsynaptic parameters studied in the previous figure. $\delta_0$ is assumed to be known. The true causal rate is defined as $g(\delta_0)/ (|\vect R| \cdot dt)$ and the estimate  $\hat{\theta}_{syn}'/(|\vect R| \cdot dt)$ for $\delta_0$. For each plot, different levels of causal rate are generated by varying $g_{0} \in \{0,0.01,0.02,...,0.1\}$  mS/cm$^2$ both for excitation ($E_{syn} = 0$ mV) and inhibition ($E_{syn} = -70$ mV). Across all panels, the empirical coverage probability equals $0.9894$. Parameters sweep left to right. \textbf{A: } $\tau_{syn} \in \{1,5,9\}$ ms. \textbf{B: } $\tau_m \in \{6,10,14\}$ ms. \textbf{C: } $\sigma_{I,post} \in \{0.5,1,2\}$ $\mu$A/cm$^2$ (referred to as $\sigma_{1}$ in the text).}
    \label{fig:lif_estimates}
\end{figure}

\begin{table}[htp]
\centering
\caption{LIF base circuit parameters}
\label{tab:lif_parameters}
\begin{tabular}{lcll}
\toprule
Parameter Name & Symbol & Unit & Value/Distribution \\
\midrule
\textbf{Cellular Properties} \\
\midrule
Membrane Capacitance & $C_m$ & $\mu$F/cm$^2$ & 1 \\
Leak Reversal Potential & $E_{\text{l}}$ & mV & -65 \\
Leak Conductance & $g_{\text{l}}$ & mS/cm$^2$ & 0.1 \\
Spike Threshold & $V_{T}$ & mV & -50 \\
Voltage Reset & $V_{R}$ & mV & $E_l$ \\
Refractory Period & $\tau_r$ & ms & 2 \\
\midrule
\textbf{Synapse} \\
\midrule
Peak Synaptic Conductance & $g_{0}$ & mS/cm$^2$ & 0.04 \\
Synaptic Reversal Potential & $E_{syn}$ & mV & 0 \\
Synaptic Time Constant & $\tau_{syn}$ & ms & 3 \\
Conduction Delay & $\tau_d$ & ms & 0 \\
\midrule
\textbf{Background Input} \\
\midrule
Input Timescale & $\tau_{I_i}$ *  & ms & 50 \\
Input Mean & $\mu_{i}$ * & $\mu$A/cm$^2$ & 0 \\
Input SD  & $\sigma_{i}$ * & $\mu$A/cm$^2$ & 1 \\
\bottomrule
* : for $i \in \{0,1,2\}$
\end{tabular}
\end{table}

\subsection{Causality and spike history in feedforward adaptive exponential (AdEx) integrate-and-fire neurons}
\label{sec:adex_main}

After developing a causal inference model in Section~\ref{sec:mono_causal_infer_model}, we proceeded to test the model in a series of numerical experiments that challenged the model's assumptions. The point process experiments of Section~\ref{sec:point_process_sim1} challenged aspects of how we constructed the background process to account for confounding; namely, the definition of $\gamma(t)$ and $\mathcal{A}.$\ref{as:4} (conditional uniformity). The LIF system experiments of Section~\ref{sec:map_to_lif}, with intrinsic dynamics and conductance-based synapses, challenged aspects of how we constructed the interaction process to account for coupling effects; namely, the definition of $\vect I^{(\vect R)}$ and $\mathcal{A}.$\ref{as:2} (timescale separation). In this final subsection, we take this challenge further and try to identify a case where the causal effect of synaptic input is perhaps more complex than a transient increase in postsynaptic spiking probability followed by exponential decay (Section~\ref{sec:map_to_lif}). 

Consider a system of AdEx model neurons~\citep{brette2005adaptive}. As before, let a presynaptic neuron ($i=0$) drive a postsynaptic neuron ($i = 1$),
\begin{align}
\label{adex_system_begin}
    C_m \frac{dV_i}{dt} &= -g_l (V_i-E_l) + g_l k_{a,i} \exp \left(\frac{V_i-V_{T,i}}{k_{a,i}} \right) - g_s g_0(V_i-E_{syn}) \mathbbm{1}\{i=1\} - I_{w,i} + I_{i} \\
    \tau_w \frac{dI_{w,i}}{dt} &= -I_{w,i} + a_i(V_i-E_l) + \mathbbm{1}\{i=1\} \sum_{y \in \mathbf{T}^{(\vect R)}} b_i \delta_d(t-y) + \mathbbm{1}\{i=0\}\sum_{r \in \mathbf{R}} b_i \delta_d(t-r) \\
    \tau_{syn} \frac{dg_s}{dt} &= -g_s + \sum_{r \in \mathbf{R}} \delta_d(t - r)
\label{adex_system_end}
\end{align}

\noindent where the LIF model has been embellished with a nonlinearity with activation slope $k_a$ and with an adaptation current $I_w$ with subthreshold adaptation coupling parameter $a$ and spike-triggered adaptation parameter $b$. A spike is triggered when $V(t)$ obtains the value $V_T + 5 k_a$ at which time the voltage $V(t)$ is as before reset to $E_l$ for a refractory period of $\tau_r$. The counterfactual interpretation of the system Eqs.~\eqref{adex_system_begin}-\eqref{adex_system_end} is exactly analogous to the LIF system of Eqs.~\eqref{eq:lif_system_begin}-\eqref{eq:lif_system_end} as already discussed noting that under the intervention $do(\vect R=\emptyset)$ the spike-triggered adaptation trigger times become $\vect T^{(\emptyset)}$. As alluded to in the previous section, here we focus on excitatory synapses only. Parameters were chosen separately for each neuron so that the presynaptic and postsynaptic cells emulate a neocortical pyramidal neuron and fast-spiking interneuron, respectively~\citep{zerlaut2018modeling}. Technically, under these parameters, the AdEx model for the postsynaptic neuron reduces to the exponential integrate-and-fire neuron (EIF), as it is a special case of the former.

As earlier with the LIF model, the synaptic conduction delay is set at $\tau_d = 0$ and Figure~{\ref{fig:adex_plot}A} plots normalized versions of $N_{S(\vect R,2\delta,\delta)} (\vect T)$ and $g(\delta)$ for $\delta \in [0,20]$ ms. The AdEx system produces an apparent non-monotonic behavior in this regime in  $g(\delta)$. This observation should be examined in the context of functional connectivity methods that use the CCG as the primary object of inference. For example, \citet{spivak2022deconvolution} argues that presynaptic autocorrelation can produce secondary oscillations in the CCG and should be corrected for by a deconvolution procedure. We plot the CCG from the simulation of Figure~{\ref{fig:adex_plot}A} in Figure~{\ref{fig:adex_plot}B}. The secondary oscillations seen here are characteristic of those thought to arise from presynaptic autocorrelation. Under the assumption $g(\delta)$ is monotonic, secondary oscillations in the CCG would indeed be an artifact manifesting when finely-timed presynaptic bursts coincide with finely-timed postsynaptic spikes arising causally from one of the presynaptic spikes in the burst (see Example~\ref{example:complex_ccg_degeneracy}). The causal postsynaptic spike then contributes to the mass of the CCG in at least two places: the large primary short-latency CCG peak~\citep{English2017} as well as in one of the secondary oscillations. Whether the causal postsynaptic spike arises from the first or second spike in the presynaptic burst dictates whether it contributes to the duplicate mass in the secondary oscillation residing in the region of negative or positive lag.

Here, we have observed that $g(\delta)$ is not monotonic, indicating, by this fact alone, that part of the secondary oscillations is causal and not an artifact due to duplicate mass in the CCG. Yet, the model neocortical pyramidal neuron does have regular bursting as well. To tease apart the contribution of each factor, we append another simulation to the AdEx system simulation as follows. Let us reuse the AdEx simulated presynaptic train $\vect R$ to keep presynaptic autocorrelation constant and reuse $\vect T^{(\emptyset)}$ to keep confounding partially constant. We take the counterfactual target spike train $\vect T^{(\emptyset)}$ and add $|\vect T| - |\vect T^{(\emptyset)}|$ spikes to it via a conditional intensity model of synaptic gain, taking the union of spikes induced by that model synapse with $\vect T^{(\emptyset)}$ (this is termed ``artificial synapse" for short). More precisely, define the synaptic gain function
\begin{align}
     &\lambda_{A}(t) | do(\vect R = \vect r) = \epsilon \int_{-\infty}^{\infty}\upsilon_0(t-\tau-\tau_d) \sum_{r \in \vect r} \delta_d(t-r) dt & \text{(generates spikes $\vect I^{(\vect R)}$)} \nonumber
\end{align}

\noindent where making a modification from before the kernel is not truncated in the direction of positive infinity: $\upsilon_0(\tau) = \exp(-\tau/\tau_s)\mathbbm{1}\{\tau \geq 0\}$. Here, $\tau_{s}$ is chosen to maximize the correlation of the resulting CCG with the CCG obtained from the initial AdEx simulation. $\epsilon$ is also chosen to produce approximately $|\vect T| - |\vect T^{(\emptyset)}|$ spikes (from the initial simulation) and then some interactions, $\vect I^{(\vect R)}$, simulated from this gain function are randomly omitted so the number of causal spikes in the second simulation exactly equal $|\vect T| - |\vect T^{(\emptyset)}|$ from the initial AdEx simulation. The resulting CCG from the artificial synapse is also displayed in Figure~{\ref{fig:adex_plot}B}.  Secondary oscillations persist in the CCG due to presynaptic autocorrelation which is confirmed by the fact that for the artificial synapse $g(\delta)$ is now monotone in Figure~{\ref{fig:adex_plot}C} as expected. However, this does not capture the whole behavior of the initial AdEx system CCG with a biophysical synapse in Figure~{\ref{fig:adex_plot}B}. This indicates that these secondary oscillations are not pure epiphenomena but instead include some causal effect that is, in fact, confounded by presynaptic autocorrelation. This was already clear by the definition of $g(\delta)$ as well, which indicates that some fraction of the causal spikes contributing to the secondary oscillations is, in fact, comprised of ``first spikes'' in response to presynaptic input (i.e., not ``second spikes'' in a rapid burst). 

While the synapse is excitatory, the non-monotonic behavior of $g(\delta)$ in the AdEx system also implies some negative gain at some points on the curve. This is likely due to a combination of the refractory periods and bias selection that causes some spikes not to occur that would have happened if the synapse had not existed. This highlights several reasons why unbiased causal effects cannot be obtained from correlation functions, including deconvolution of the CCG with the presynaptic auto-correlogram (ACG) outside neatly controlled cases. Furthermore, it must be stressed that there might exist many other causes, including network oscillations, that give rise to secondary oscillations in CCG, and so spiking correlation functions, in general, fail to address the fundamental problem of causal inference: confounding.

Estimation of the AdEx system ensues exactly as before. In Figure~{\ref{fig:adex_plot}D-F} point and interval estimates are plotted for all simulations just explored using eleven equally-spaced values for $g_0 \in [0,0.1]$ mS/cm$^2$ to generate different levels of causal rate. Figure~{\ref{fig:adex_plot}A} displays estimates for the full AdEx system defining $\theta_{syn}$ as $g(\delta_0)$. While all the confidence intervals still cover the true parameter, there is a clear underestimation. However, Figure~{\ref{fig:adex_plot}E} is obtained from the artificial synapse with the same presynaptic spike trains as Figure~{\ref{fig:adex_plot}D} and the bias vanishes. Thus, we may deduce that the bias observed in Figure~{\ref{fig:adex_plot}A} is not due to presynaptic autocorrelation. This is expected, as no assumptions were made about $\vect R$ in the theoretical development of the monosynaptic causal inference model. One possibility is that the monosynaptic causal inference model does not best approximate this dynamical system using $\delta_0$. Instead, we tried using $\delta_1 = \argmax_{\delta} g(\delta)$. While the resulting estimates are not as precise as in the LIF, it appears in this model that  $g(\delta_1)$ is better identified than $g(\delta_0)$ for most coupling strengths as shown in Figure~{\ref{fig:adex_plot}F}. As a tangential point, this also shows that estimation, in general, might be reasonable across some range of $\delta$.

\begin{figure}
    \centering
\includegraphics[width=.8\textwidth]{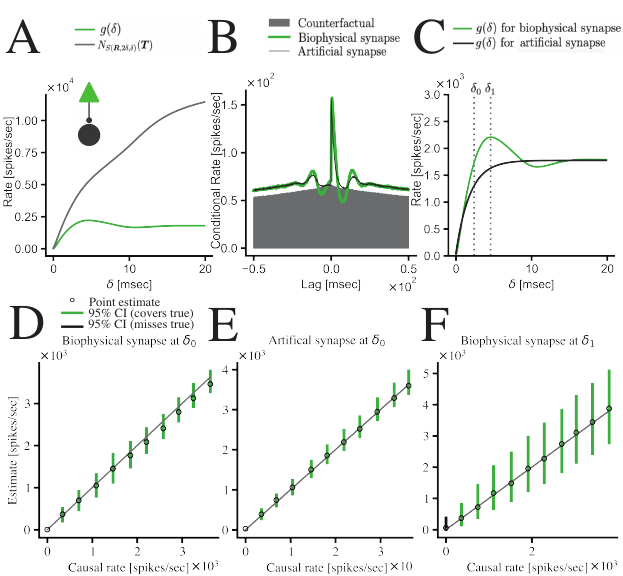}
    \caption[Causality and spike history effects in an AdEx system of a neocortical pyramidal cell driving a fast-spiking interneuron]{{\bf Causality and spike history effects in an AdEx system of a neocortical pyramidal cell driving a fast-spiking interneuron.} The model neurons exhibit strong spike history effects. \textbf{A: } $g(\delta)$ and $N_{S(\vect R,2\delta,\delta)}(\vect T)$ are plotted for the AdEx system in a simulation lasting 27.77 simulated hours. Note the non-monotone fluctuations in $g(\delta)$. \textbf{B: } The green CCG is the observation $\chi(\vect R,\vect T)$ from the same simulation. The gray-filled CCG is $\chi(\vect R,\vect T^{(\emptyset)})$ from the corresponding frozen noise simulation with the synapse removed. The black line is a CCG constructed by adding synchronous spikes to $\vect T^{(\emptyset)}$ (termed ``artificial synapse'') such that the total spike count equals $|\vect T|$; the synchronous spike times are added with a time constant chosen to maximize correlation with $\chi(\vect R,\vect T)$. Note that secondary oscillations persist in the black CCG due to presynaptic autocorrelation, but the full behavior remains unexplained by assuming independence of the causal spikes in the postsynaptic train. \textbf{C: } $g(\delta)$ for $\vect R$ and $\vect T$ is plotted in green. Plotted in black are $g(\delta)$ for $\vect R$ and the modified target train constructed by adding the artificial synapse to $\vect T^{(\emptyset)}$. Note the non-monotone fluctuations vanish.  \textbf{D: } Estimates are biased for $g(\delta_0)$. \textbf{E: } The bias vanishes with the artificial synapse. \textbf{F: } Bias is also slightly reduced by estimating $g(\delta)$ at $\delta_1 = \argmax_{\delta} g(\delta)$, however perhaps at the cost of less precision and a Type 1 error for $g(\delta_1) = 0$. Different levels of causal rate are generated by varying $g_0 \in \{0,.01,...,0.1\}$ mS/cm$^2$ as before.}
    \label{fig:adex_plot}
\end{figure}

\begin{table}[htp]
\centering
\caption{AdEx Circuit Parameters}
\label{tab:paired_adex_parameters}
\begin{tabular}{lcll}
\toprule
Parameter Name & Symbol & Unit & Value \\
\midrule
\textbf{Cellular Properties} \\
\midrule
Leak Conductance & $g_{\text{l}}$ & mS/cm$^2$ & 1/15 \\
Refractory Period & $\tau_r$ & ms & 5 \\
Membrane Capacitance & $C_m$ & $\mu$F/cm$^2$ & 1 \\
Leak Reversal Potential & $E_{\text{l}}$ & mV & -65 \\
Voltage Reset & $V_{\text{R}}$ & mV & $E_l$ \\
Adaptation Current Timescale & $\tau_w$ & ms & 500 \\
Spike Threshold & $V_{T}$ & mV & -50 \\
\midrule
\textbf{Synapse} \\
\midrule
Peak Synaptic Conductance & $g_0$ & mS/cm$^2$ & 0.05 \\

Synaptic Reversal Potential & $E_{\text{syn}}$ & mV & 0 \\
Synaptic Time Constant & $\tau_{\text{syn}}$ & ms & 3 \\
Conduction Delay & $\tau_d$ & ms & 0 \\
\midrule
\textbf{Pyramidal Neuron} \\
\midrule
Activation Slope & $k_{a,0}$  & mV & 2 \\
Adaptation Conductance & $a_0$ & mS/cm$^2$ & 2.04 \\
Adaptation Increment & $b_0$ & $\mu$A/cm$^2$ & 0.02 \\
Reset Condition & $V_{T, 0}$ & mV & $V_{T} + 5k_{a,\text{0}}$ \\
\midrule
\textbf{Interneuron} \\
\midrule
Activation Slope & $k_{a,1}$ & mV & 0.5 \\
Adaptation Conductance & $a_{1}$ & mS/cm$^2$ & 0 \\
Adaptation Increment & $b_{1}$ & $\mu$A/cm$^2$ & 0 \\
Reset Condition & $V_{T, 1}$ & mV & $V_{T} + 5k_{a,\text{1}}$ \\
\midrule
\textbf{Background Input Currents} \\
\midrule
Input timescales & $\tau_{I,i}$ * & ms & 50 \\
Input Mean & $\mu_{i}$ * & $\mu$A/cm$^2$ & 0 \\
Input SD & $\sigma_{i}$ * & $\mu$A/cm$^2$& 1 \\
\bottomrule
* : for $i \in \{0,1,2\}$
\end{tabular}
\end{table}

\newpage

\section{Neural perturbations for testing assumptions and fitting free parameters}

The frequently invoked separation of timescales hypothesis in monosynaptic inference~\citep{Csicsvari1998,Amarasingham2006,English2017,ren2020model,spivak2022deconvolution} to some degree suggests we may learn something useful by studying a toy model of instantaneously coupled Bernoulli processes in discrete time. Importantly, this setting possesses the feature that presynaptic and postsynaptic spikes can be thought of as sequences of binary treatment and outcome variables. When the synaptic effect is very fine-timescale, as is often observed \textit{in vivo}~\citep{platkiewicz2021monosynaptic}, and when firing rates are sparse, this might be a reasonable approximation. Of course, the analogy breaks in obvious ways including long synaptic decay times, temporal summation of PSPs, spike history effects, etc. But the toy model can clarify issues about causality and, fortunately for neuroscience, well-developed causal inference concepts for binary treatment and outcomes variables can then be applied to pairwise spike trains in a fairly straightforward way. In this section, we make this simplification to discuss how perturbation experiments (e.g., optogenetics) could test the monosynaptic model's assumptions or fit free parameters.

\subsection{Monosynaptic model calibration in an ideal neural perturbation experiment}
\label{sec:ideal_experiment}

For simulations in this setting, we will also retreat back to point process simulations that are even simpler than the one of Section \ref{sec:point_process_sim1}.
As building blocks, piecewise constant excitability functions will be used for various purposes,

\begin{equation}
    b_i(t_j) = \sum_{k \in \mathbbm{Z}^*} m_{i,k}\mathbbm{1}\{(k-1)\Delta \leq t_j < k \Delta\}, \text{ for } i = 0,1
\end{equation}

\noindent where the $m_{i,k}$ are repurposed from a multivariate skew of dimension $n=2$ (see Eq.~\ref{eq:multivaraiateskew}), with discrete time points $\{t_j\}_{j \in \{0,1,..., D\}}$ for an experiment of duration $D$, and where $\Delta$ is the bandwidth of the amplitudes chosen as a constant equal to the statistical free parameter of the same name defined in previous sections. This will be used to construct conditional intensity functions in an idealized monosynapse model. Working in discrete time, sets of spike times in this section will be defined as sets of integers.

Here the relationship between $\theta_{syn}$ and the \textit{probabilities of causation} of \citet{tian2000probabilities} is demonstrated in simulation. This provides an alternative set of assumptions to identify causal effects that utilize observational and experimental data. As before we work in the toy case of instantaneously coupled Bernoulli processes in simulation. For $0 \leq j \leq D$, define the conditional intensity functions,
\begin{align}
    &\lambda_R(t_j) | b_0(t) = 
    \rho_0 b_0(t_j) \label{eq:toy_rate_model1_begin} \\
    &\lambda_{T}(t_j) | b_1(t), do(\vect R = \vect r) = \rho_1 b_1(t_j) + \epsilon \mathbbm{1} \{ t_j \in \vect r\}  \label{eq:toy_rate_model1_end}
\end{align}
\noindent where $\rho_0$ and $\rho_1$ are normalization factors and $\epsilon$ is a fixed instantaneous coupling constant chosen such that $\lambda_{T}(t_j) | \vect R$ remains a proper intensity function. We map this toy model onto the monosynaptic causal model by identifying $\vect R$ and $\vect T^{(\vect R)}$ as the sets of spike times generated from $\lambda_R | b_0(t)$, $\lambda_{T}(t_j) | b_1(t), \vect R$ respectively. 

We implement the model structurally by extending the analogy of Example~\ref{rem:measure_spike_model}, with $\vect R$ and $\vect T$ conditionally independent given knowledge of the (causal) conditional intensity functions (\ref{eq:toy_rate_model1_end}). Expanding on that, define an idealized neural intervention of the presynaptic neuron as one that causally induces a new reference train $\vect r_0$; $do(\vect R = \vect r_0)$ is implemented by independently sampling the reference train from a constant intensity function $\lambda_{opto}(t_j) = \lambda_0$ inducing an experimental version of the postsynaptic intensity, $\lambda_T(t_j) | b_1(t), do(\vect R = \vect r_0) =  \rho_1 b_1(t_j) + \epsilon \mathbbm{1} \{ t_j \in \vect r_0\}$ with outcomes $\vect T^{(\vect r_0)}$. Here (in discrete time) this is effectively the common notion of experimental randomization whereby every time bin is assigned to spike by mechanisms that act independently and homogeneously across time.

Returning to \textit{probabilities of causation}, in the general case of Bernoulli random variables $X$ and $Y$, respectively, \citet{tian2000probabilities} define these probabilities as follows,
\begin{align}
    PN &= \mathbbm{P}\left( Y^{(X=0)} = 0 | X = 1, Y = 1\right) && \text{(probability of necessity)} \label{eq:pn}\\
    PS &= \mathbbm{P}\left(Y^{(X=1)}  = 1 | X =0, Y = 0 \right) && \text{(probability of sufficiency)} \\
    PNS &= \mathbbm{P}\left(Y^{(X=1)} = 1, Y^{(X=0)} = 0 \right) && \text{(probability of necessity \& sufficiency)}.
    \label{eq:pns}
\end{align}

For example, probability of necessity (PN) is the probability that $\{X=1\}$ is a necessary cause of the effect $\{Y=1\}$. It is the probability that, given the event that $\{X=1\}$ and $\{Y=1\}$ both occur, $Y$ is 0 when $X$ is forced (via intervention) to be 0. More loosely, $X$ would be 0, were it not that $Y$ is 1; that is, $X$ is the {\it necessary} cause of $Y$. PS and PNS have similar interpretations. We refer the interested reader to Ch. 9 of \citet{pearl2009} for a fuller review.

To map these probabilities into our experiments (e.g., spikes simulated from the structural causal model in the specification above, including Eqs.~\eqref{eq:toy_rate_model1_begin}-\eqref{eq:toy_rate_model1_end}), let $V$ be a random time: $V \sim Uniform\{1,2,...,D\}.$ Then $X := \ind \{ V \in \vect R\},$ and $Y := \ind \{ V \in \vect T\}$ and apply Eqs. (\ref{eq:pn}-\ref{eq:pns}).
Hence, in simulation, we will identify the ground truth of PN with its intervention-inferred numerical estimate $PN = \theta_{syn}/|\vect R \cap \vect T|$. This is the true proportion of causal synchrony to observed synchrony. (Note that the noise processes are not iid, so there is an additional, implicit assumption that the noise processes are mixing quickly enough to make the error in this identification negligible. We do not analyze this error, or incorporate a variability assessment.)  

Pearl identifies several ways to identify $PN$ from observational and experimental data~\citep{pearl2009}. For our purposes, an acceptable assumption is \textit{monotonicity}, which here simply requires a synapse to be strictly excitatory ($\epsilon > 0$) or strictly inhibitory ($\epsilon < 0$). Let us explain the excitatory case. The inhibitory case follows precisely the same logic but redefines the outcome variable as silence rather than a spike. We follow \citet{pearl2009} and for finite data assert by hypothesis an alternative estimate for $PN = \theta_{syn}/|\vect R \cap \vect T|$ as,
\begin{align}
     \hat{PN}_{exp} &:=   
      \bigg(\frac{|\vect T \cap \vect R|}{D}\bigg)^{-1}\bigg(\frac{|\vect T|}{D} - \frac{|\vect T^{(\vect r_0)} \setminus \vect r_0|}{D-|\vect r_0|}\bigg), \text{ if } \epsilon \geq 0
    \label{eq:PN_syn}
\end{align}
where as defined earlier $D$ is the duration of the experiment. The estimator uses spontaneous and perturbation data as just outlined. Under the monosynaptic causal inference model, the analogous estimator is denoted  $\hat{PN}_{obs} = \hat{\theta}_{syn}/|\vect R \cap \vect T|$ requiring only observational data under its assumptions. Likewise, we suggest $PNS = \epsilon$ in the toy model with the alternative estimator,
\begin{align}
      \hat{PNS}_{exp} &:= 
     \frac{| \vect T^{(\vect r_0)} \cap \vect r_0 |}{|\vect r_0|} - \frac{|\vect T^{(\vect r_0)} \setminus \vect r_0|}{D-|\vect r_0|}, \text{ if } \epsilon \geq 0.
    \label{eq:PNS_syn}
\end{align}
\indent The monosynaptic causal inference model's corresponding estimate will be $\hat{PNS}_{obs} = \hat{\theta}_{syn}/|\vect R|$. Notice this gives a more principled account of what neurophysiologists often call \textit{efficacy}~\citep{levick1972lateral} or \textit{spike transmission gain}~\citep{abeles1991corticonics}, which are, loosely, the excess probability of a postsynaptic spike given that a presynaptic spike occurred. We use the word \textit{loosely} because the word \textit{excess} has no universal interpretation (see excellent review in~\citet{stevenson2023circumstantial}), and to our knowledge, none have formally interpreted \textit{excess} in terms of counterfactuals and potential outcome random variables.  Finally, for the ground truth numerical \textit{probability of sufficiency} let $\iota \coloneqq (\mathbbm{Z}^* \cap [0,D)) \setminus (\vect R \cup \vect T)$. Then, $PS = |\vect T^{(\iota)} \cap \iota| /\iota$ with alternative estimate,
\begin{align}
     \hat{PS}_{exp} &:= 
      \bigg(\frac{D-|\vect T \cup \vect R|}{D}\bigg)^{-1}\bigg(\frac{|\vect T^{(\vect r_0)} \cap \vect r_0|}{|\vect r_0|} - \frac{|\vect T|}{D}\bigg), \text{ if } \epsilon \geq 0
    \label{eq:PS_syn}
\end{align}

\noindent and estimated from observational data only by the monosynaptic causal inference model as,

\begin{equation}
\hat{PS}_{obs} = \bigg( \frac{D - |\vect T \cup \vect R|}{D}\bigg)^{-1}\bigg(\hat{PNS}_{obs} - \bigg(\frac{|\vect T \cap \vect R|}{D}\hat{PN}_{obs}\bigg)\bigg).
\end{equation}

As mentioned before, these quantities can be obtained for inhibition in the exact same way where the queried postsynaptic outcome variable is silence. For visualization purposes, in the inhibitory case, we define the probabilities of causation through multiplication by $-1$ so that an estimate of inhibition can be plotted simultaneously with excitation and compared with $\theta_{syn}$ on its negative support. For example, in the inhibitory case, we will have,

\begin{equation}
    \hat{PN}_{syn} = -1 \bigg(\frac{|\vect R \setminus \vect T|}{D}\bigg)^{-1} \bigg( \frac{D - |\vect T|}{D} - \frac{D-|\vect T^{(\vect r_0)} \cup \vect r_0|}{D-|\vect r_0|}\bigg), \epsilon < 0.
\end{equation}

The significant observation is that the probabilities of causation are obtained from experimental and observational data without appeal to some of the assumptions that make $\theta_{syn}$ identifiable from observational data alone. Namely, $\mathcal{A}.$\ref{as:4}  (conditional uniformity) and $\mathcal{A}.$\ref{as:2} (separation of timescales) are not required to identify the probabilities of causation. For this reason, if these idealized concepts could be extended to fit more realistic aspects spike trains recorded $\textit{in vivo}$, we have here provided an experimental test of $\mathcal{A}.$\ref{as:4} and $\mathcal{A}.$\ref{as:2} that could be conducted in the laboratory. Essential in this endeavor would be confidence limits, say for $PN$, with finite data, which is research currently being pursued~\citep{li2022learning}. However, the final section of this study will argue that such an experiment is not easily achieved by current experimental technologies (e.g., optogenetic stimulation). These alternative estimators also might provide a route to estimate the free parameter $\delta$, $\tau$, and $\Delta$.

We conclude this section by simulating spike trains from the toy model in Eqs.~\eqref{eq:toy_rate_model1_begin}-\eqref{eq:toy_rate_model1_end}. Simulation details are exactly analogous to those in Figure \ref{fig:point_process_D}. Figure~{\ref{fig:idealized_experiment}} shows the results of forty-two simulations (twenty-one for excitatory and inhibitory estimates) as just described and each plot shows the corresponding point estimates for $PNS$ (Figure~{\ref{fig:idealized_experiment}A}), $PS$ (Figure~{\ref{fig:idealized_experiment}B}), and $PN$ (Figure~{\ref{fig:idealized_experiment}C}). In each case, a tight correspondence is shown between the $\theta_{syn}$-derived estimates, which come from observational data, and the alternative estimates, which use a combination of experimental and observational data.

\newpage

\begin{figure}
    \centering
\includegraphics[width=.99\textwidth]{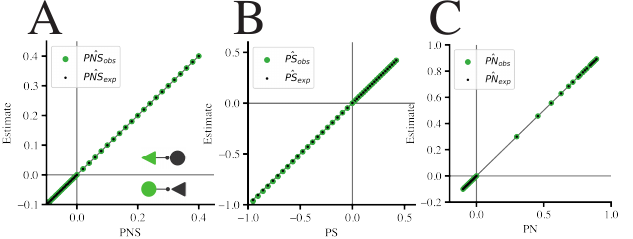}
    \caption[An idealized experimental test of the monosynaptic model's assumptions]{{\bf An idealized experimental test of the monosynaptic model's assumptions.} In sparse firing conditions where presynaptic and postsynaptic spikes can be approximated as binary treatment and outcome variables, the monosynaptic model can be related to the \textit{probabilities of causation} of  \citet{tian2000probabilities} in a toy spiking model. $\hat{PNS}_{exp}$, $\hat{PS}_{exp}$, and $\hat{PN}_{exp}$  are derived directly from \citet{tian2000probabilities} and provide alternative estimates for monosynaptic causal effects by combining neural intervention data with spontaneous data and thus require fewer assumptions on the background processes, providing an experimental test of the monosynaptic causal inference model's assumptions. $\hat{PNS}_{obs}$, $\hat{PS}_{obs}$, and $\hat{PN}_{obs}$ are different normalizations of $\hat{\theta}_{syn}$ obtained from observational data as described in the text. \textbf{A: } Probability of necessity and sufficiency ($PNS$) corresponds to $\epsilon$ in the toy model. \textbf{B: } Probability of sufficiency $(PS)$. \textbf{C: } Probability of necessity ($PN$).}
    \label{fig:idealized_experiment}
\end{figure}

\subsection{Even strong perturbations might quite strongly fail as randomized experiments.}
\label{sec:opto_fail}
In the previous section, we explored conceptually the notion of an ideal neural intervention in a toy spiking model. The purpose of this was to highlight connections between $\theta_{syn}$ and more well-established causal inference concepts and to speculate 
about avenues for future research that might make the interventions suitable for more realistic dynamical models. Another concern persists, which is the degree to which current experimental technologies actually achieve the theoretical notion of an \textit{intervention} in causal inference. Recently, \citet{lepperod2023inferring} fruitfully analyzed the confounding that arises from optogenetic stimulation activating many neurons that may be unobserved. However, while it is well-understood that stimulation often increases the empirical rate of the presynaptic neuron on a coarse timescale~\citep{English2017}, it is not clear in a dynamical system how much deconfounding occurs at the level of voltage and hence what the proper interpretation of juxtacellular or optogenetic stimulation is. In this section, we caution against simple interpretations (and thus show the difficulty of obtaining the ideal intervention we proposed in the previous section) by injecting stochastic input currents into correlated but unconnected LIF neurons, as well as stimulating them with a biophysically detailed channelrhodopsin model~\citep{williams2013computational}. 

A simple example to consider is two unconnected LIF neurons with common input that produces structure in the CCG. An ideal intervention, where each point in time is randomly assigned a presynaptic spike or not (see Section~\ref{sec:ideal_experiment}) should destroy all structure in the cross-correlogram during stimulation, yielding a flat histogram. Consider two LIF neurons,
\begin{equation}
    C_m \frac{dV_i}{dt} = -g_l(V_i-E_l) + I_c(t) + I_i(t) + I_{p}(t)\mathbbm\{i=1\}, \text{ for } i = 0,1 
\label{eq:lif_system_opto}
\end{equation}
\noindent where as before if at time $t_0$, $ V(t_0) = V_{T}$, the voltage is reset to $E_l$. With the same form as Eq.~\ref{eq:ou_process} but with a slight change in notation, $I_c(t)$ is common OU noise to both neurons, $I_i(t)$ for $i=0,1$ is independent OU noise for each neuron. $I_p(t)$ is either an injected current identical to the stimulus to be described momentarily or the same stimulus filtered by the channelrhodopsin (ChR2) model of~\citet{williams2013computational}. 

Let $S(t)$ be the stochastic stimulus, then
\begin{align}
    &I_p(t) = S(t) && \text{(for current stimulation)} \\
    &I_p(t) = g_{ChR2}G(V)(O_1 + \gamma O_2) (V_0-E_{ChR2}) && \text{(for optogenetic stimulation)} \label{eq:opsin_current}
\end{align}
\noindent where in the notation of~\citet{williams2013computational} $g_{ChR2}$ is the max conductance of the photocurrent, $E_{ChR2}$ is the reversal potential for channelrhodopsin, $G(V)$ is a voltage-dependent rectification function, $O_1,O_2$ are open state probabilities, and $\gamma$ is a normalization factor. Eq.~\eqref{eq:opsin_current} is identical to Eq. 1 in ~\citet{williams2013computational}, and we replace what in their notation is termed $S_0(\theta)$ in their Eq. 11 with our stimulus $S(t)$. We refer the reader to the rest of that study since the channelrhodopsin model is rather complicated and has various state variables and parameters. We used identical parameters from the original study for channelrhodopsin.

In simulation, we take $S(t)$ to be a special discrete construction of a Gauss-Markov process with Hurst or Hölder parameter $H \in [0,1]$. The motivation for this is to generate repeatable spike patterns in the presynaptic neuron regardless of the level of other sources of noise~\citep{taillefumier2014transition}, constituting the notion of experimental intervention. The parameter $H$ plays the same role as in fractional Brownian motion, intuitively describing how rough (small $H$) versus smooth (high $H$) the trajectory is, however the process used here is colored (i.e., its power spectrum is not flat). The process was developed as an injected current in previous work to suggest that more reliable spiking patterns can be induced into a LIF neuron to the degree that $H$ is small regardless of the neuron's level of independent noise~\citep{taillefumier2008haar}. The construction of the process is described in Appendix~\ref{app:ou_process}. In this setting, it is simply being employed as technology to produce reliable spiking responses to stimulation~\citep{mainen1995reliability,levi2022error}, although the tenability of this very statement in this setting is what is being tested in the simulation. Consider that if a spiking pattern were perfectly reliable to a repeated stimulus, then an experimentalist would know that they are deconfounding in the sense of the $do(\cdot)$ operator of causal inference.

Figure~{\ref{fig:biophysical_perturbation}} simulates the system in Eq.~\eqref{eq:lif_system_opto} for different parameters of the stochastic input current or light stimulus; the timescale $\tau_{H}$ and Hölder parameter $H$. In each simulation, the timescales of the intrinsic processes $I_0(t)$, $I_1(t)$, and $I_c(t)$ were set to $10$ ms, and their amplitudes to unit variance. The amplitudes of the stimulations were then adjusted so that the reference neuron's empirical rate during stimulation was approximately $470 \%$ greater than the spontaneous rate as in the experiment of~\citet{English2017}. Equalizing firing rate across experimental conditions in this way, surprisingly quite strong common input correlations persist for current injection and optogenetic stimulation. Furthermore, varying the input parameters $H$ and $\tau_H$ leads to hardly detectable differences in the deconfounding as measured through the CCG. If current or optogenetic stimulation fulfilled the notion of $do(\cdot)$ as applied to spike trains in Section~\ref{sec:ideal_experiment}, the CCG during stimulation should be flat. 

\begin{figure}
    \centering
\includegraphics[width=.65\textwidth]{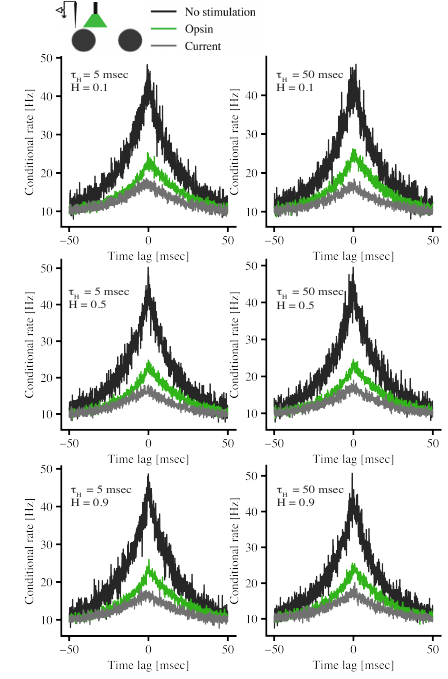}
    \caption[Strong juxtacellular or optogenetic stimulation might fail to be randomized experiments]{{\bf Strong juxtacellular or optogenetic stimulation might fail to be randomized experiments.} Two LIF neurons are driven by common inputs, and their CCG is plotted (black plots). With the same frozen noise input, the model system is subjected to either (1) current injection in the pattern of a Gaussian process (gray plots) or (2) photocurrent stimulation in the same pattern on a biophysically detailed opsin model affixed to one of the LIF neurons (green plots). A special Gaussian process is utilized from a theory that predicts reliable spike patterns to be produced to the degree that a parameter $H \in [0,1]$ is small. The timescale of the stimulus, $\tau_H$, is also split into two conditions, 5 ms and 50 ms simulations. The variance of the stimulations was adjusted for electric current injection or photostimulation such that the stimulated firing rate was approximately $470 \%$ greater than the spontaneous rate~\citep{English2017}. Even for this strong perturbation, confounding common synaptic input correlations persist and do not significantly differ given the character of the input when the variance of the stimulation is adjusted to produce equal firing rates across conditions.}
    \label{fig:biophysical_perturbation}
\end{figure}

\newpage

\section{Appendix}

\begin{table}[htp]
\centering
\caption{List of abbreviations}
\label{tab:abrev_descriptions}
\begin{tabular}{ll}
\toprule
\textbf{Abbreviation} & \textbf{Definition} \\
\midrule
ACG & Auto-correlogram \\ 
AdEx & Adaptive exponential integrate-and-fire neuron \\ 
CCG & Cross-correlogram \\ 
\textit{cdf} & Cumulative distribution function \\ 
CGF & Cumulant generating function \\ 
ChR2 & Channelrhodopsin-2 \\ 
DC & Direct convolution \\
DFT/IDFT & Discrete-Fourier transform and its inverse \\
FFT/IFFT & Fast-Fourier transform and its inverse \\
\textit{iid} & Independent and identically distributed \\
INT & Interneuron \\
LIF & Leaky integrate-and-fire neuron \\ 
OU & Ornstein–Uhlenbeck process \\
\textit{pmf} & Probability mass function \\
PSP & Postsynaptic potential  \\
PYR & Pyramidal neuron \\
SD & Standard deviation  \\
\bottomrule
\end{tabular}
\end{table}

\subsection{Examples of confounding and non-identifiability in the CCG}
\label{app:ccg_exam}

To prepare for examples that demonstrate this issue, let $\vect R$ and $\vect T$ be a finite set of spike times (a point process) for an experiment of fixed duration. We will, in general, consider a reference spike train $\vect R$ hypothesized to be presynaptic, and a target spike train $\vect T$, hypothesized to be postsynaptic. A goal is to quantify the evidence for that hypothesis and for a number of its characteristics. We are thus interested in the potential outcome random variable, $\vect T^{(\vect R=\vect r)}$, abbreviated $\vect T^{(\vect r)},$ which is the target train, in a causal model, induced by $do(\vect R = \vect r)$. That is, we are interested in the causal influence of $\vect R$ on $\vect T.$ For any spike trains $\vect X$ and $\vect Y$ define the unnormalized sample cross-correlation function (sample CCF) as,
\begin{align}
    \hat{\chi}(\vect X, \vect Y,\tau) &\coloneqq \int_{-\infty}^{\infty} \sum_{x \in \vect X} \delta_d(t-x) \sum_{y \in \vect Y} \delta_d(t-y + \tau)  \; \; dt
\end{align}

\noindent where $\delta_d$ is the Dirac delta function. We will also write $\chi(\vect X, \vect Y,\tau) = \E[\hat{\chi}(\vect X, \vect Y,\tau)]$ and will occasionally assume the spike trains are discrete, reinterpreting the notation accordingly when specified. The term unnormalized cross-correlogram (CCG) likewise refers to a binned version of $\hat{\chi}(\vect X, \vect Y,\tau)$.

The following examples motivate the approach of this article. {Figure~\ref{fig:neural_interventions}} illustrates their simulation and the causal decompositions described in the examples. {Example~\ref{rem:measure_spike_model}} presents an example of a causal model in terms of point process models, and subsequent simulations will utilize this definition of causality. 

We start with the simplest model one might imagine.
\begin{example}[Instantaneously-coupled Bernoulli processes with fixed coupling constant $\epsilon$]
\label{rem:measure_spike_model}
Define a probability space which contains $\vect \omega = (\omega_1,...,\omega_{2N}),$ a vector of $2N$ independent uniform [0,1] random variables.  Then consider the following potential outcomes model: $\vect R (\vect \omega) = \{t_j : \omega_j \leq \lambda_{R}\}$ and $\vect T^{(\vect R=\vect r)}(\vect \omega) = \{t_j : \omega_{j+N} \leq \lambda_{ T} + \epsilon \mathbbm{1} \{ t_j \in \vect r\}\}.$ By independence, there is no confounding (of $\vect R$ and $\vect T$). The average causal effect of the coupling at time $t$, $E[ \ind\{ t \in \vect T\} - \ind\{ t \in \vect T^{(\vect R=\emptyset)}\} ],$ is $\epsilon$. Consider, for example, the intervention $do(\vect R = \emptyset)$: $\vect T^{(\vect R=\emptyset)}(\vect \omega) = \{t_j : \omega_{j+N} \leq \lambda_{T} + \epsilon \mathbbm{1} \{ t_j \in \emptyset\}\} = \{t_j : \omega_{j+N} \leq \lambda_{T} \}$. $T^{(\vect R=\emptyset)}(\vect \omega)$ is defined as a function on the same probability space as the functions $\vect R (\vect \omega)$ and $\vect T(\vect \omega)=\vect T^{(\vect R)}(\vect \omega)$. 
$(\omega_1, \omega_2, ..., \omega_{2N})$ are so-called `background' variables. We think of a particular realization of $\vect \omega$ as encoding the state(s) of the `external' world. Interventions modify the relations between $\vect R$ and $\vect T$ to define potential outcomes for $\vect T$, given that the state(s) of the world (i.e., the background variables $\vect \omega$) are fixed (i.e., `frozen') over potential outcomes. 
\end{example}


{Example~\ref{example:simple_ccg_degeneracy}} now examines the behavior of the CCF for {Example~\ref{rem:measure_spike_model}} demonstrating that $\epsilon$ is not identifiable.

\begin{example}[Identical CCGs with different coupling strength]
\label{example:simple_ccg_degeneracy}
One can verify from the independence relations that a normalized CCF for the model in Example~\ref{rem:measure_spike_model} is $\chi(\vect R, \vect T^{(\vect R)},\tau)/N =  \E[\sum_{t=1}^N \mathbbm{1}\{t-\tau \in \vect R\} \mathbbm{1}\{t\in \vect T^{(\vect R)}\}]/N = \lambda_R(\lambda_T + \epsilon) \mathbbm{1}\{\tau=0\} + (\lambda_R\lambda_T + \lambda_R^2\epsilon) \mathbbm{1}\{\tau \neq 0\}$ dismissing edge effects. Since in this model the CCF is flat everywhere but the coupling lag at $\tau = 0$, the CCF peak, $\rho = \chi(\vect R, \vect T^{(\vect R)},\tau=0)-\chi(\vect R, \vect T^{(\vect R)},\tau = z)$ for any $z \neq 0$, can be related to the average causal effect as $\epsilon = \rho/(\lambda_R-\lambda_R^2)$ where we have normalized the peak by $\lambda_R$ which is standard in functional connectivity studies. $\epsilon$ and $\rho$ are not equal. Consider two situations. In \textbf{Situation A} the model has parameters $\lambda_{R,A}, \epsilon_{A}, \lambda_{T,A}$. In \textbf{Situation B}, an identical $\lambda_R$-normalized CCF is obtained by setting $\lambda_{R,B} = 1 - \epsilon_{A}, \epsilon_{B} = 1 - \lambda_{R,A}, \lambda_{T,B} = \lambda_{T,A} + \lambda_{R,A}\epsilon_{A} - \lambda_{R,B}\epsilon_{B}$. For example, we can have the average causal effects be $\epsilon_{A} = 0.2,$ in Situation A, and $\epsilon_{B} = 0.8,$  in Situation B, and yet their CCFs are identical.
\end{example}

In the previous example, $\epsilon$ is identifible if supplemented by one additional unknown, $\lambda_R$. Here that is trivial if the process is stationary, however neural data is known to be highly nonstationary~\citep{Softky1993,shinomoto1999ornstein} leading to extreme difficulty in estimating analogous time-varying quantities~\citep{Amarasingham2015}. When $\lambda_R$ is unknown, one can verify through level set analysis that the model parameters in the example can vary widely for fixed $\epsilon$, suggesting large relative bias if $\lambda_R$ is slightly misestimated, even when $\lambda_R$ is small. 

Related to the observation that $\lambda_R$ confounds estimation in {Example~\ref{example:simple_ccg_degeneracy}}, it has long been understood that presynaptic autocorrelation might in some way influence the CCG between two neurons~\citep{moore1970statistical} leading to suggestions that deconvolution of the CCG with the presynaptic ACG  might help in deconfounding, particularly under stationarity assumptions~\citep{spivak2022deconvolution}. In the next example, we examine the nature of this confounding in a nonstationary setting. Two neurons are given confounding oscillatory backgrounds. In addition, the presynaptic cell emits a burst of three spikes approximately every second. As a thought experiment, imagine these bursting events alternate such that the spike times in the bursts are either generated by a Gaussian with small variance or large variance. Causal conclusions from the CCG vary widely in their dependence on whether the causal interactions tend to occur among the bursts with small or large variances, highlighting the non-identifiability of causal inference from correlation functions altogether. Later, we use this intuition to construct confidence intervals by supposing causal events occur at these limiting cases of presynaptic firing, thus bounding the estimate over this uncertainty (Section~\ref{sec:excitatory_ci}).

\begin{example}[Different CCGs with identical coupling strength]
\label{example:complex_ccg_degeneracy}
Consider the following generative model for an experiment of duration $D$. Let background intensity functions be $\lambda_{R}(t) = \lambda_{T}(t) = \alpha \cos(\omega t) + \alpha$ for $t \in [0,D), \alpha \in [0,1/2)$ both generating sets of real-valued points $\vect R_{0}$ and $\vect T_{0}$, respectively. Define a sequence of latent events $0 \leq \ell_1, \ell_2,...,\ell_K \leq D$ such that $\ell_k - \ell_{k-1} \sim Uniform(0.8,1.2)$ seconds. Let these latent events be the center of a burst of three spikes, $X_{k,i} \sim \mathcal{N}(l_k,\sigma_A \mathbbm{1}\{k \text{ is even}\} + \sigma_B \mathbbm{1}\{k \text{ is odd}\})$ for $i \in \{1,2,3\}$. Collect these events in a set $\vect R_1 = \cup_{k \in \mathbbm{N}} \cup_{i=1}^{3} X_{k,i}$ and define the presynaptic spike train as $\vect R = \vect R_0 \cup \vect R_1$. Also, let $Y_{k,i} \sim \mathcal{N}(X_{k,i} + d,\sigma_s \mathbbm{1}\{k \text{ is even}\} + \sqrt{2\sigma_A^2 + \sigma_s^2} \mathbbm{1}\{k \text{ is odd}\})$ for $i \in \{1,2,3\}$ and $\vect X = \cup_{\{k \text{ even}\}} \cup_{i=1}^{3} Y_{k,i}$ and $\vect Y = \cup_{\{k \text{ odd}\}} \cup_{i=1}^{3} Y_{k,i}$. Now consider two situations. In \textbf{Situation A}, $\vect T^{(\emptyset)} = \vect T_0$ and $\vect T^{(\vect R)} = \vect T_0 \cup \vect X$. Since cross-correlation is a linear operator and the constituent processes are in superposition, the unnormalized CCF can be expressed in approximate closed-form as  $\chi_A(\vect R,\vect T, \tau) = \chi(\vect R_0, \vect T_0,\tau) + \chi(\vect R_1, \vect T_0,\tau) + \chi(\vect R_0, \vect X,\tau) + \chi(\vect R_1, \vect X,\tau) \approx D \alpha^2/2\cos(\omega t) + \alpha^2 + D^{-1}(\E[|\vect R_1|]\E[|\vect T_0|] + \E[|\vect X|]\E[|\vect R_0|]) + \frac{1}{2}\E[|\vect R_1|] \mathcal{N}(\tau |d,\sigma_s) + \frac{3-1}{2}\E[|\vect R_1|] \mathcal{N}(\tau |d,\sqrt{2\sigma_A^2 + \sigma_s^2})$ where we have dismissed edge effects. In \textbf{Situation B}, $\vect T^{(\emptyset)} = \vect T_0$ and $\vect T^{(\vect R)} = \vect T_0 \cup \vect Y$ and by the same logic $\chi_B(\vect R, \vect T,\tau) \approx D \alpha^2/2\cos(\omega t) + \alpha^2 + D^{-1}(\E[|\vect R_1|]\E[|\vect T_0|] + \E[|\vect Y|]\E[|\vect R_0|]) + \frac{1}{2}\E[|\vect R_1|] \mathcal{N}(\tau |d,\sqrt{2\sigma_A^2 + \sigma_s^2}) + \frac{3-1}{2}\E[|\vect R_1|] \mathcal{N}(\tau |d,\sqrt{2\sigma_B^2 + 2\sigma_A^2 + \sigma_s^2})$. In both cases, we write $\frac{3-1}{2}\E[|\vect R_1|]$ to highlight that
there are three causal events per burst, we subtract one as it has already been counted in the penultimate term, and such a scenario occurs for half of the presynaptic bursts contributing to the unnormalized CCF. If $\sigma_A$ and $\sigma_S$ are both small and not appreciably different and $\sigma_A << \sigma_B$, the penultimate term of $\chi_A(\vect R,\vect T, \tau)$ will appear as a monosynaptic feature and the penultimate term of $\chi_B(\vect R,\vect T, \tau)$ will appear as a background feature, leading to very different causal conclusions in both situations although the causal effects - in the sense of $|\vect X|/|\vect R| = |\vect Y|/|\vect R|$ - and presyaptic spike trains are exactly equal in each situation, demonstrating that mere knowledge of the CCG and presynaptic ACG would be insufficient for causal inference. 
\end{example}

These examples are intentionally dramatic to be instructive and often exceed plausible neurophysiological behavior. Namely, the firing rates are often quite high, and we used Gaussian functions for analytic tractability, although they allow for some causality to occur in reverse time. However, if the examples are understood in mathematical detail, it's straightforward to see that the degeneracy is quite general, especially if we are concerned with the relative error of estimates. Moreover, the regime of nonstationary, high presynaptic bursting coinciding with information transfer is perhaps the most biologically relevant~\citep{Ostojic2015,Herrmann2002,pressley2011dynamics}. A common way presynaptic bursting manifests in the CCG is as secondary oscillations visually distinct from the primary monosynaptic peak because of the refractory periods ~\citep{spivak2022deconvolution}. We explore this in Section~\ref{sec:adex_main} and Figure~\ref{fig:adex_plot}. Note, however, that in Example~\ref{example:complex_ccg_degeneracy} the influence of bursting on the CCG depends on how the temporal resolution of bursting interacts with the temporal resolution of the causal interactions they associate with; a concept that should generalize beyond the idealizations made here. So it is not guaranteed that refractory periods will dissolve the issue. In fact, statistical dependence between the temporal resolution of bursting and the causal interactions (e.g., from short-term plasticity)  might likely make the interpretation even more subtle, leading to multiple sources of estimation error at different lags of the CCG (that is, a combination of Example~\ref{example:simple_ccg_degeneracy} and Example~\ref{example:complex_ccg_degeneracy}). Complex dependencies between background input currents and other forms of nonstationarity suggest that the toy examples here may paint a forgiving picture of the confounding present in real neural data~\cite {mehler2018lure}.

\subsection{Monosynaptic causal inference model proofs}
\label{app:proof}

\mytheoremone*

\begin{proof}  
\leavevmode\newline
\\
\textbf{Case 1}: Excitation, $\theta_{syn} \geq 0$.

\begin{quote}
\noindent For an arbitrary $k \in \mathbbm{Z}^*$ we have from the definition of the model Eq.~\eqref{eq:synon} we have, for all $\vect r$,
\begin{align} 
N_{\gamma(k\Delta) \cap S(\vect r)}\big(\vect T^{(\vect r)}\big)
&= N_{\gamma(k\Delta)}\big(\vect I^{(\vect r)} \setminus \bigcup_{r \in \vect r} \{S(r) : N_{S(r)}(\vect B) > 0 \}\big)  +  N_{\gamma(k\Delta) \cap S(\vect r)}\big(\vect B \big)
\end{align}
\noindent where $S(\vect r)$ has been excluded from the subscript of the increment in the first term of the RHS by $\mathcal{A}.$\ref{as:2}. Taking conditional expectations with respect to $(\vect \Gamma (\vect T), \vect R)$, and applying $\mathcal{A}.$\ref{as:1} (consistency) and linearity, we have 
\begin{align} \label{eq:forBsub}
&\E\bigg[N_{\gamma(k\Delta) \cap S(\vect R)}\big(\vect T\big) \bigg| \vect \Gamma (\vect T), \vect R  \bigg]   \\
&= \E \bigg[ N_{\gamma(k\Delta)}\big(\vect I^{(\vect r)} \setminus \bigcup_{r \in \vect R} \{S(r) : N_{S(r)}(\vect B) > 0 \}\big)\bigg| \vect \Gamma (\vect T), \vect R \bigg] \\
& \quad +  \E \bigg[ N_{\gamma(k\Delta) \cap S(\vect R)}\big(\vect B\big) \bigg| \vect \Gamma (\vect T), \vect R \bigg] . \nonumber
\end{align}

\noindent By $\mathcal{A}.$\ref{as:4} (conditional uniformity), 
\begin{equation}
    \E[ N_{\gamma(k\Delta) \cap S(\vect R)}\big(\vect B\big) | \vect \Gamma (\vect T), \vect R] = q(\vect R, k\Delta)\E[ N_{\gamma(k\Delta)} (\vect B) | \vect \Gamma (\vect T), \vect R].
\end{equation}

Furthermore, Eq.~\eqref{eq:synon} under $\mathcal{A}.$\ref{as:2} gives $N_{\gamma(k\Delta)}\big(\vect I^{(\vect r)} \setminus \cup_{r \in \vect r} \{S(r) : N_{S(r)}(\vect B) > 0 \}\big) = N_{\gamma(k\Delta)}(\vect T^{(\vect r)}) - N_{\gamma(k\Delta)}(\vect T^{(\emptyset)})$, for all $\vect r$. Eq.~\eqref{eq:forBsub} then becomes,
\begin{align}
&\E\bigg[N_{\gamma(k\Delta) \cap S(\vect R)}\big(\vect T\big) \bigg| \vect \Gamma (\vect T), \vect R \bigg]   \\ \label{eq:choosedelneed}
 &= \E\bigg[N_{\gamma(k\Delta)}(\vect T )  - N_{\gamma(k\Delta)}(\vect T^{(\emptyset)}) \bigg| \vect \Gamma (\vect T), \vect R \bigg] +  q(\vect R, k\Delta) \E\bigg[ N_{\gamma(k\Delta)}( \vect B ) \bigg| \vect \Gamma (\vect T), \vect R \bigg] \\
 &= \E\bigg[N_{\gamma(k\Delta)}(\vect T )  - N_{\gamma(k\Delta)}(\vect T^{(\emptyset)}) \bigg| \vect \Gamma (\vect T), \vect R \bigg]\\
 & \quad \quad +  q(\vect R, k\Delta) \E\bigg[ N_{\gamma(k\Delta)}\big(\vect T\big) - \bigg(N_{\gamma(k\Delta)}(\vect T) - N_{\gamma(k\Delta)}(\vect T^{(\emptyset)})\bigg) \bigg| \vect \Gamma (\vect T), \vect R \bigg]
\end{align}

\noindent where the substitution of $N_{\gamma(k\Delta)} (\vect B)$ inside the expectation of the last term again results from Eq.~\eqref{eq:synon} under $\mathcal{A}.$\ref{as:2}. 
Rearranging, we obtain
\begin{align}
&\E\bigg[N_{\gamma(k\Delta)}(\vect T) - N_{\gamma(k\Delta)}(\vect T^{(\emptyset)}) \bigg| \vect \Gamma (\vect T), \vect R \bigg] \\
& \quad \quad =\E\bigg[\frac{N_{\gamma(k\Delta) \cap S(\vect r)}\big(\vect T \big) -  q(\vect R, k\Delta) N_{\gamma(k\Delta)}\big(\vect T\big)}{1- q(\vect R, k\Delta)} \bigg| \vect \Gamma (\vect T), \vect R \bigg],
\end{align}
using the general fact that $\E[f(Y)\E[X|Y]] = \E[\E[f(Y)X|Y]] = \E[Xf(Y)]$ for (measurable) random variables $X, Y,$ and functions $f$. Summing over $k$ and taking expectations on both sides of the above gives

\begin{align}
\E\bigg[ \hat{\theta}_{syn} \bigg] &= \E\bigg[ \sum_{k \in \mathbbm Z^*} \frac{N_{\gamma(k\Delta) \cap S(\vect r)}\big(\vect T \big) -  q(\vect R, k\Delta) N_{\gamma(k\Delta)}\big(\vect T\big)}{1- q(\vect R, k\Delta)}  \bigg] \label{eq:thetaexcite}\\
&= \E\bigg[ \sum_{k \in \mathbbm Z^*} N_{\gamma(k\Delta)}(\vect T) - N_{\gamma(k\Delta)}(\vect T^{(\emptyset)}) \bigg] \\
&= \E\bigg[ N_{S(\vect R)}(\vect T)  - N_{S(\vect R)}(\vect T^{(\emptyset)}) \bigg] = \theta_{syn}.
\end{align}
\end{quote}

\textbf{Case 2}: Inhibition, $\theta_{syn} \leq 0$.

\begin{quote}
    By Eq.~\eqref{eq:synon} and $\mathcal{A}.$\ref{as:2} the following invariant holds for all realizations of the inhibitory model and every $\vect r$,
\begin{align}
    N_{\gamma(k\Delta) \setminus S(\vect r)}(\vect T^{(\vect r)}) = N_{\gamma(k\Delta) \setminus S(\vect r) }(\vect T^{(\emptyset)}).
\end{align}

In the inhibitory model, Eq.~\eqref{eq:synon} under $\mathcal{A}.$\ref{as:2} gives $N_{\gamma(k\Delta) \setminus S(\vect r) }(\vect T^{(\emptyset)}) = N_{\gamma(k\Delta) \setminus S(\vect r) }(\vect B), \forall \vect r$. Using this substitution and applying similar steps as in Case 1, including $\mathcal{A}.$\ref{as:1} and $\mathcal{A}.$\ref{as:4}, we have
\begin{align}
&\E \bigg[N_{\gamma(k\Delta) \setminus S(\vect R)}(\vect T) \bigg| \vect R, \vect \Gamma (\vect T) \bigg] = \bigg(1 - q(\vect R, k\Delta)\bigg) \E \bigg[ N_{\gamma(k\Delta)}(\vect B) \bigg| \vect R, \vect \Gamma (\vect T) \bigg]. \label{eq:inhib1} 
\end{align}

In the inhibitory model, we again have from Eq.~\eqref{eq:synon} under $\mathcal{A}.$\ref{as:2} $N_{\gamma(k\Delta)}(\vect T^{(\vect r)}) = N_{\gamma(k\Delta)}(\vect B) - N_{\gamma(k\Delta)}(\vect B \cap \cup_{r \in \vect r} \{S(r) : N_{S(r)}(\vect I^{(\vect r)}) > 0 \})$. Using this relation to substitute $N_{\gamma(k\Delta)}(\vect B)$ in Eq.~\eqref{eq:inhib1} and again using $\mathcal{A}.$\ref{as:1} for terms inside the expectation,
\begin{align}
&\E \bigg[N_{\gamma(k\Delta) \setminus S(\vect R)}(\vect T) \bigg| \vect R, \vect \Gamma (\vect T) \bigg] \\
&= \bigg(1 - q(\vect R, k\Delta)\bigg) \E \bigg[ N_{\gamma(k\Delta)}(\vect T) + N_{\gamma(k\Delta)}(\vect B \cap \bigcup_{r \in \vect R} \{S(r) : N_{S(r)}(\vect I) > 0 \})  \bigg| \vect R, \vect \Gamma (\vect T) \bigg] \label{eq:inhib2} \\
&= \bigg(1 - q(\vect R, k\Delta)\bigg) \E \bigg[ N_{\gamma(k\Delta)}(\vect T) - \bigg(N_{\gamma(k\Delta)}(\vect T^{(\vect R)}) - N_{\gamma(k\Delta)}(\vect T^{(\emptyset)}) \bigg) \bigg| \vect R, \vect \Gamma (\vect T) \bigg]
\end{align}

\noindent where the change of sign inside the expectation of the last line results from the fact that under $\mathcal{A}.$\ref{as:2} $N_{\gamma(k\Delta)}(\vect T^{(\vect R)}) - N_{\gamma(k\Delta)}(\vect T^{(\emptyset)}) \leq 0$ in the inhibitory model. Rearranging and following analogous steps as used in Case 1,

\begin{align}
\theta_{syn} = -\sum_{k}\E\bigg[\frac{N_{\gamma(k\Delta) \setminus S(\vect R)}\big(\vect T\big) -  (1-q(\vect R, k\Delta)) N_{\gamma(k\Delta)}\big(\vect T\big)}{1- q(\vect R, k\Delta)}\bigg] \label{eq:inhibest}.
\end{align}

Setting this expression equal to Eq.~\eqref{eq:thetaexcite} all terms cancel.
\end{quote}
\end{proof}

\lemmaone*
\begin{proof}
We will establish (\ref{eq:l_minus}). (\ref{eq:argmax_plus}) can be established in the same way.
\begin{align}
&\mathbbm{P}\bigg(N_{S(\vect R)} (\vect B)  \leq  c \bigg| \vect q ( \vect R, \vect T), \vect J \bigg) \\
&=\sum_{n=0}^{\text{c}} \sum_{Q\in \vect J^{[\text{n}]}}^{} \prod_{i \in Q}^{} q(\vect R, T_i) \prod_{k\in (\vect J \setminus Q)} (1-q(\vect R, T_k))&&\\ \nonumber
   &=\sum_{n=0}^{\text{c}-1} \sum_{Q\in \{\vect J \setminus x\}^{[n]}}^{} \prod_{i \in \{Q\cup x\}}^{} q(\vect R, T_i) \prod_{k\in (\vect J \setminus (Q \cup x))}^{}(1-q(\vect R, T_k)) \\
   &\quad \quad + \sum_{n=0}^{\text{c}} \sum_{Q\in \{\vect J \setminus x\}^{[n]}}^{} \prod_{i \in Q}^{} q(\vect R, T_i) \prod_{k\in \{\vect J \setminus Q\}}^{}(1-q(\vect R, T_k)), &&
\end{align}

\noindent using $\{ Q \in \vect J^{[n]} : x \in Q \} = \{ Q \cup x : Q \in \{\vect J \setminus x\}^{[n-1]}\}$. For an arbitrary $x \in \vect J$,  this implies
\begin{align}
    &\frac{\partial}{\partial q(\vect R, T_x)} \mathbbm{P}\bigg(N_{S(\vect R)} (\vect B) \leq  c \bigg| \vect q ( \vect R, \vect T), \vect J \bigg) =&& \\ &=\sum_{n=0}^{\text{c}-1} \sum_{Q\in \{\vect J \setminus x\}^{[n]}}^{} \prod_{i \in Q}^{} q(\vect R, T_i) \prod_{k\in \vect J \backslash (Q \cup x)}^{}(1-q(\vect R, T_k)) \nonumber \\    
    &\quad \quad - \sum_{n=0}^{\text{c}} \sum_{Q\in \{\vect j \setminus x\}^{[n]}}^{} \prod_{i \in Q}^{} q(\vect R, T_i) \prod_{k\in \vect J \setminus (Q \cup x)} (1-q(\vect R, T_k))&& \\
    &= - \sum_{Q\in \{\vect J \setminus x\}^{[c]}}^{} \prod_{i \in Q}^{} q(\vect R, T_i) \prod_{k\in \vect J \setminus (Q \cup x) }(1-q(\vect R, T_k)) \leq 0,
\label{eq:cdf_deriv_neg}
\end{align}

\noindent because $q(\vect R,T_i) \in [0,1]$ for all $i \in \vect K$. Note that $\partial/\partial q(\vect R,T_x) [\mathbbm{P}\left(N_{S(\vect R)} (\vect B) \leq  c| \vect q ( \vect R, \vect T), \vect J \right)]$ is independent of $q(\vect R, T_x).$


(Proof by contradiction.) Suppose $\vect J_{\theta_{syn}}^- \not\in \argmin_{\vect j \in \mathcal D} \{ c^{-}( \vect q( \vect R, \vect T), \vect j) \}$. Fix any $\vect J^* \in \argmin_{\vect j \in \mathcal D} c^{-}( \vect q( \vect R, \vect T), \vect j).$ Accordingly, assume there exists an $x^* \in \vect J_{\theta_{syn}}^-$ such that $x^* \not\in \vect J^*$ so that $q( \vect R, T_{x^*}) < q( \vect R, T_m)$ for some $m \in \vect J^*$ by the definition of $\vect J_{\theta_{syn}}^-.$ (If such an $x^*$ does not exist, then $c^{-}( \vect q(\vect R, \vect T), \vect J^*) = c^{-}( \vect q(\vect R, \vect T), \vect J^-_{\theta_{syn}})$ by construction.) By Eq.~\ref{eq:cdf_deriv_neg}, this implies $c^{-}( \vect q(\vect R, \vect T), \vect J^*) > c^{-}(\vect q(\vect R, \vect T), \vect J^* \cup \{x^*\} \setminus \{m\} ).$ This is a contradiction as $\vect J^* \cup \{x^*\} \setminus \{m\} \in \mathcal D.$

\end{proof}

\lemmatwo*

\begin{proof} The proof is by induction. The case $n=1$ is self-evident. Observe that
\begin{align}
    \Prob \left( \sum_{i=1}^n X_i \leq k | Z \right) &= \Prob \left( \sum_{i=1}^{n-1} X_i \leq k \bigg| Z, X_n=0 \right) ( 1- \Prob( X_n=1 | Z )) \\
    & + \Prob \left( \sum_{i=1}^{n-1} X_i \leq k-1 \bigg| Z, X_n=1 \right) \Prob( X_n=1 ) 
\end{align}
and
\begin{equation}
    \Prob \left( \sum_{i=1}^n Y_i \leq k \right) = \Prob \left( \sum_{i=1}^{n-1} Y_i \leq k \right) + \Prob \left( \sum_{i=1}^{n-1} Y_i \leq k-1 \right) p_n.
\end{equation}

\noindent Conditioning on $X_n = 0$, the induction hypothesis is satisfied for $X_1,X_2,...,X_{n-1}$. Therefore,

\begin{equation}
     \Prob \left( \sum_{i=1}^{n-1} X_i \leq k | Z, X_n=0 \right) \leq \Prob \left( \sum_{i=1}^n Y_i \leq k \right)
\end{equation}

\noindent and analogously,

\begin{equation}
\Prob \left( \sum_{i=1}^{n-1} X_i \leq k \bigg | Z, X_n = 1 \right) \leq \Prob \left(\sum_{i=1}^{n-1}Y_i \geq k-1\right).
\end{equation}

\noindent Also note,

\begin{equation}
    \bigg\{\sum_{i=1}^n Y_i \leq k\bigg\} \subseteq \bigg\{\sum_{i=1}^n Y_i \leq k - 1 \bigg\} \implies \Prob \left(\sum_{i=1}^{n-1} Y_i \leq k \right) \leq \Prob \left(\sum_{i=1}^{n-1} Y_i \leq k-1 \right)
\end{equation}

\noindent and,

\begin{equation}
\Prob(X_n = 1) = \sum_{_i X} \Prob(X_n=1|Z,_i X) \Prob(_i X) \leq \sum_{_i X} p_n \Prob(_i X) = p_n.
\end{equation}
From this reasoning we obtain,
\begin{align}
    \Prob \left(\sum_{i=1}^n Y_i \leq k \right) &= \Prob \left(\sum_{i=1}^{n-1} Y_i \leq k \right) \left( 1-p_n \right) 
    + \Prob \left(\sum_{i=1}^{n-1} Y_i \leq k - 1 \right) p_n \\
    &\geq \Prob \left(\sum_{i=1}^{n-1} Y_i \leq k \right) \left(1-\Prob(X_n=1)\right)
    + \Prob \left(\sum_{i=1}^{n-1} Y_i \leq k-1 \right) \Prob(X_n=1) \\
    &\geq \Prob \left(\sum_{i=1}^{n-1} X_i \leq k \bigg| Z, X_n = 0\right) \left(1-\Prob(X_n=1)\right)
    + \Prob \left(\sum_{i=1}^{n-1} X_i \leq k-1 \bigg | Z, X_n = 1\right) \Prob(X_n=1) \\
    &= \Prob \left(\sum_{i=1}^{n} X_i \geq k \bigg| Z \right).
\end{align}
\end{proof}

\propone*

\begin{proof}
\noindent We will prove the first inequality Eq.~\ref{left-tail-bd}. Eq.~\ref{right-tail-bd} can be proved in the same way.

Note that
\begin{equation}
\Prob \bigg ( N_{S(\vect R)} (\vect T) - h \leq  c^{-}( \vect q( \vect R, \vect T ), \vect J) \bigg |  q( \vect R, \vect T ), \vect J \bigg )  \leq \alpha/2, 
\end{equation}
by the definition of $c^-(\cdot),$ and symmetry. Thus, by the fact that $\vect J \in \mathcal{D}(\vect R, \vect T)$, under $H_0$ and by Lemma \ref{argmin-lemma},
\begin{equation}
    c^-( \vect q( \vect R, \vect T ), \vect J_h^-) \leq c^-( \vect q( \vect R, \vect T ), \vect J) \label{prop-key-ineq}
\end{equation}
for all realizations of the model. We have,
\begin{flalign}
    \Prob \bigg ( N_{S(\vect R)} (\vect T) - h &\leq  c^-( \vect q( \vect R, \vect T ), \vect J^-_h) \bigg |  q( \vect R, \vect T ), \vect J \bigg ) &    \\
    & \leq \Prob \bigg ( N_{S(\vect R)} (\vect T) - h \leq  c^-( \vect q( \vect R, \vect T ), \vect J) \bigg |  q( \vect R, \vect T ), \vect J \bigg )  \leq \alpha/2,
\end{flalign}
(almost surely). Therefore 
\begin{flalign}
    \Prob \bigg( N_{S(\vect R)} (\vect T) - h &\leq  c^-( \vect q( \vect R, \vect T ), \vect J^-_h) \bigg) \\
    &= \mathbb{E} \bigg[ \Prob \bigg( N_{S(\vect R)} (\vect T) - h \leq  c^-( \vect q( \vect R, \vect T ), \vect J^-_h) \left|  q( \vect R, \vect T ), \vect J \bigg) \right. \bigg] \\
    &\leq \alpha/2.
\end{flalign}
\end{proof}

\begin{remark}
We found the following intuition useful regarding the conditional inference in Proposition \ref{prop-critregion}. It is not necessary that the critical region's boundary term, $c^{-}( \vect q( \vect R, \vect T), \vect J^{-}_{h}),$ be $(\vect q( \vect R, \vect T), \vect J)$-measurable, because of the inequality Eq.~\ref{prop-key-ineq}. Yet, a key feature is that it is $(\vect R, \vect T)$-measurable, so that, speaking informally, its evaluation does not require ``knowledge'' of $\vect J.$
\end{remark}

\subsection{Construction of multivariate Ornstein-Uhlenbeck process}
\label{app:ou_process}

To try to create reliable spike patterns, we utilize the Ornstein-Uhlenbeck construction of~\citet{taillefumier2008haar} as an injected current or light stimulus. We describe a multivariate version of the process because it is far more efficient to sample a multivariate version and then concatenate the dimensions into a longer trial. The process used for the former task is easily recovered in the one-dimensional case.  These authors construct the Ornstein-Uhlenbeck process $U(t)$ from discrete Haar-like basis functions. We focus on the discrete representation since the primary goal is to simulate the process. The whole process is divided into dyadic segments, with the following basis functions tiling the dyadic segments at various resolutions,
\begin{equation}
\lambda_{n,k}(t) \coloneqq
\left\{\!\begin{aligned}
&\frac{\sigma_H \cdot \sinh{(|\alpha_H|(t-2k\cdot2^{-n}))}}{\sqrt{\alpha_H \sinh{(\alpha_H2^{1-n})}}} && \text{if } (2k)2^{-n} \leq t < (2k+1)2^{-n}\\[1ex]
&\frac{\sigma_H \cdot \sinh{(|\alpha_H|(2(k+1)2^{-n}-t))}}{\sqrt{\alpha_H \sinh{(\alpha_H2^{1-n})}}} && \text{if } (2k+1)2^{-n} \leq t < 2(k+1)2^{-n}\\[1ex]
&0 && \text{otherwise,}
\end{aligned}\right.
\end{equation}

\noindent for $0 \leq 2k < 2^n$ with timescale $\tau_H$, scaling parameters $\sigma_H$ and $\alpha_H$, and where

\begin{equation}
    \lambda_{0,0}(t) = \frac{\sigma_H \cdot \exp(-\alpha_H/2)\sinh(|\alpha_H|t)}{\sqrt{\alpha_H \cdot \sinh{(\alpha_H)}}}.
\end{equation}

A parameter $H \in [0,1]$ describes how the amplitude of these basis functions scale with the resolution of the support, $\Delta_H$. For a multi-dimensional version, for $1 \leq i \leq 2n-1$ and $1 \leq j \leq n$, define the matrices,

\begin{align}
\vec{\vect M}_n(i,j) &= \mathbbm{1}\{i=2j-1\} + \mathbbm{1}\{i/2=j\} + \mathbbm{1}\{i/2= j-1\}\\
\vec{\vect U}_n(i,j) &= \mathbbm{1}\{i \text{ is even} \} \\
\vec{\vect V}_n(i,j) &= \mathbbm{1}\{i \text{ is odd} \}.
\end{align}

For a multivariate process of dimension $M$ and $2^N$ time points, let the value of the process at each time point be $\vec{x}_t \in \mathbbm{R}^m$. For a resolution $n$, arrange the values of the process at the dyadic points $(0 \cdot 2^N/2^{n},0 \cdot 2^N/2^{n},...,2^{n} \cdot 2^N/2^{n})$ into distinct matrices indexed by $n$,

\begin{align}
\vec{\vect X}_{n} = \begin{bmatrix}
\vec{x}_{0 \cdot 2^N/2^{n}}  \\
\vec{x}_{1 \cdot 2^N/2^{n}}  \\
\vdots  \\
\vec{x}_{2^{n} \cdot 2^N/2^{n}}
\end{bmatrix}.
\end{align}

For numerical reasons, the endpoints are clamped such that $\vec{\vect X}_0 = [0]_{2,m}$ and the full process $\vec{\vect X}_N$ can then be computed for $0 \leq n \leq N-1$ efficiently by the recursive matrix operations,

\begin{align}
\overbrace{\begin{bmatrix}
\vec{x}_{0 \cdot 2^N/2^{n+1}}  \\
\vec{x}_{1 \cdot 2^N/2^{n+1}}  \\
\vec{x}_{2 \cdot 2^N/2^{n+1}}  \\
\vdots  \\
\vec{x}_{2^{n+1} \cdot 2^N/2^{n+1}} \\
\end{bmatrix}}^{\displaystyle \vec{\vect X}_{n+1}} &= \bigg(\frac{1}{2}\cosh^{-1}{(\Delta_H/2\tau_H)} \vec{\vect U}_{n+1} + \vec{\vect V}_{n+1} \bigg) \odot \vec{\vect M}_{n+1} \overbrace{\begin{bmatrix}
\vec{x}_{0 \cdot 2^N/2^{n}}  \\
\vec{x}_{1 \cdot 2^N/2^{n}}  \\
\vec{x}_{2 \cdot 2^N/2^{n}}  \\
\vdots  \\
\vec{x}_{2^{n} \cdot 2^N/2^{n}} \\
\end{bmatrix}}^{\displaystyle \vec{\vect X}_{n}}\\
&+ \alpha_H [\tau_H \tanh{(\Delta_H/2\tau_H)}]^H \sqrt{\frac{1}{\tau_H}} 
\underbrace{\begin{bmatrix}
\vec{0}_m \\
\vec{\xi}_{1 \cdot 2^N/2^{n+1}} \\
\vec{0}_m  \\
\vec{\xi}_{3 \cdot 2^N/2^{n+1}} \\ 
\vdots \\
\vec{\xi}_{(2^{n+1}-1) \cdot 2^N/2^{n+1}} \\
\vec{0}_m \\
\end{bmatrix}}_{\displaystyle \vec{\vect \Xi}_{n+1}}  \label{eq:sample_mvoup}
\end{align}

\noindent where for fixed $n$ $\vec{\vect \Xi}_{n}$ contains in its odd rows $iid$ normal random vectors, $\vec{\xi}_t$, obeying $\mathcal{N}(\vec{0}_m,\vect \vec{\Sigma}_{m,m}(n))$ and zeros otherwise. That is, each resolution $n$ has a characteristic dependence structure parameterized by the covariance matrices $\vec{\vect \Sigma}_{m,m}(n)$ for $1 \leq n \leq N$. $\vec{\vect \Sigma}_{m,m}(0)$ is not defined since the endpoints are clamped. For background inputs into the HH-type model system, we used a constant covariance matrix for all $n$. As was previously done in the text, the Vine Beta method, with its parameter fixed to $0.1$, was used to generate a random covariance matrix with strong positive and negative associations in three dimensions. Sampling is made much more efficient by the following. First, for a process of dimension $M$ sample~Eq.\eqref{eq:sample_mvoup} instead for a process of $M\cdot N_{trial}$ with the covariance matrix for the $M$-dimensional process replaced with the block matrix $\vec{\vect \Sigma}_{N_{trial}m,N_{trial}m}(n) = \vec{\vect I} \otimes \vect \Sigma_{m,m}(n)$ where $\otimes$ is the Kronecker product and $\vec{\vect I}$ the identity matrix. Implementing~Eq.\eqref{eq:sample_mvoup} recursively is efficient up to matrices of moderate size, and finally, then every $j$-th column for $j \in \{1,2,...m\}$ of $\vect X_{2^N}$ can be concatenated together $N_{trial}$  times.

When we used the above process as a stimulus in the simulated neural perturbation experiments, both $H$ and $\tau_H$ were varied as indicated in Section~\ref{sec:opto_fail}.

\subsection{Rationale for monosynaptic confidence interval algorithm}
\label{app:algorithm_rationale}

Synaptic inference is commonly performed on low-dimensional and easy-to-visualize objects, the CCG in particular, and the most immediate objection to the theory outlined in previous sections is that it may be computationally prohibitive. Point estimation via Eq.~\eqref{eq:pointestimator} is quite simple, but for confidence intervals, this objection may be reasonable because inference is performed on the sequence of spike counts and synchrony in small temporal intervals, which is of much higher dimension and grows with the duration of the spike trains. Earlier, this objection was neutralized via a principled algorithm. Here, we describe the algorithms in prose and highlight their rationale.

Given an observation $\vect R$ and $\vect T$ all the probabilities in Eq.~\eqref{eq:rankord_inhib} can be obtained as well as $\Tilde{\vect J}^{-}_h$ and $\Tilde{\vect J}^{+}_h$ for any $h$. As mentioned in the main text, a naive but easy-to-grasp strategy to compute confidence intervals would be to begin with the two-tailed hypothesis $H_0: \theta_{syn}=0$ enacted through convolution of the distributions with success probabilities $q(\vect R,Z_i)$ for $i \in \Tilde{\vect J}^{-}_0$ and then evaluation of tail areas given the observed value of $N_{S(\vect R)}(\vect T)$. Then, if we reject the null hypothesis at the upper tail proceed to test positive values for $H_0: \theta_{syn}=h$ in the sequence $h \in (1,2,3,...)$ recomputing the tail probability each time for distributions arising from $\Tilde{\vect J}^{-}_h$ until we fail to reject. If, instead, we fail to reject the null hypothesis at the left tail, proceed to test negative values $h \in (-1,-2,-3,...)$ until we fail to reject. The process then needs to be repeated for $\Tilde{\vect J}^{+}_h$ except for the case $h=0$. Each value $h$ tested is expensive because we must compute a sum of independent random indicators and a central question is how to reuse computations most effectively across these tests. For most data of reasonable size a much faster strategy is to apply standard binary search~\citep{lehmer1960teaching} to the sequence $(0,1,2,...,N_{S(\vect R)}(\vect T))$ (i.e., the values of $h$ to test) with the search query being the location of adjacent values in the sequence for which one and not the other fails to reject the null hypothesis $H_0: \theta_{syn} = h$.

The next question is how exactly to convolve the distributions that emerge from every step of binary search so that the tail area can be evaluated; that is, we must compute the \textit{cdf} of a sum of independent but not necessarily identically distributed indicators. While in practice this is often computed with the \textit{Fast-Fourier transform} (FFT), FFT can have very large relative errors for small tail probabilities~\citep{keich2005sfft}. In contrast, \textit{direct convolution} (DC) is the most accurate method and uses only the convolution definition of the distribution function of a sum of random variables. While its runtime complexity is $O(N^2)$, typically motivating the use of FFT, \citet{biscarri2018simple} observe that DC is still most efficient when convolving a small number of vectors or many vectors of small dimensions (e.g., Bernoulli vectors) suggesting DC and FFT can work in concert in a divide-and-conquer scheme. Through both theoretical considerations and experimentation, we adopt a similar mixed approach, also taking inspiration from~\citet{peres2021exactly}. These studies address the general problem of convolving independent indicators, while the problem here is to construct confidence intervals, and there are various other domain-specific constraints that we can exploit. For example, a special feature of our problem is we may assume many of the random indicators to be convolved will have equal success probabilities~\citep{jeck2015closed}. 

The overview is as follows. First, as a first pass, we compute a conservative confidence interval using Chernoff bounds to estimate tail probabilities at every iteration of binary search. Then, in a second pass of binary search, we refine the Chernoff confidence interval on a much-narrowed search space by computing tail probabilities with a mix of DC and FFT accompanied by a method for recovering the relative accuracy of FFT via exponential tilting~\citep{wilson2016accurate}. The strategy also exploits redundant structure in various ways to minimize computations. We digress to quickly introduce Chernoff's bound to those unfamiliar.

\begin{remark}
The following is a well-known result. Let $X_i \sim Be(p_i)$ for $i \in \{1,2,...n\}$. Consider the condition the sum is above some bound $t$,
\begin{align}
    \sum_i X_i  \geq t.
\end{align}
\noindent By multiplying through by a constant $\lambda$, exponentiating, and applying Markov's inequality the generic Chernoff bound is obtained,
\begin{align}
\label{eq:genecher}
    P\bigg(\sum_i X_i  \geq t\bigg) \leq exp(-\lambda t)\E[exp(\lambda \sum_i X_i)].
\end{align}
\noindent In the case of sums of independent indicators a bit more work can yield the result,
\begin{align}
    P\bigg(\sum_i X_i  \geq t\bigg) \leq (\mu e/t)^te^{-\mu}
\end{align}
\noindent such that $ln(t/\mu) = \lambda$ and $\mu = \sum_i p_i$~\citep{vershynin2018high}.
\end{remark}

Denote $\mathcal{\Tilde{C}}_{CF}$ as confidence intervals for $\theta_{syn}$ obtained by using Chernoff's bound as a tail probability estimate. If we substitute computation of the exact \textit{cdf} with Chernoff's bounds, it is guaranteed that $\mathcal{\Tilde{C}}_{CF}\supseteq \mathcal{\Tilde{C}}$. Furthermore, unlike the exact \textit{cdf}, iterative computation of Chernoff's bound is very cheap for binary search on the sequence $(1,2,...,N_{S(\vect R)}(\vect T))$ relative to its imprecision. 

After $\mathcal{\Tilde{C}}_{CF}$ is calculated, the search space is significantly reduced and we can now choose a more accurate method to implement a second pass of binary search on the sequence $(CF_{L},CF_{L}+1,...,CF_{U}-1,CF_{U})$ where $CF_{L}$ and $CF_{U}$ are the lower and upper confidence bounds for $\theta_{syn}$ obtained via Chernoff's bound. The first pass with Chernoff's bound also guarantees that we need the lower confidence interval to be at least $CF_{L}$, meaning we are assured that at least $CF_{L}$ random variables \textit{will not} need to be convolved in the second pass. On the other hand, the first pass assures that at least $N_{S(\vect R)}(\vect T) - CF_{U} + |\vect T \setminus S(\vect R)|$ random variables \textit{will} need to be convolved and hence we can design an algorithm that only computes those convolutions once then reuses the result. 

We will describe the process for computing exact confidence intervals that ensues for the case of $h>0$ and for the tail corresponding to $\Tilde{\vect J}^{+}_h$. The other cases are analogous, and we will conclude by highlighting alterations. For the remainder of this section, we will describe the strategy in words to supplement the description with rationale. Where ambiguity might be present here, the reader can refer to the precise description as summarized in Algorithm~\ref{alg:convolution} and Algorithm~\ref{alg:confidence}.

Given $CF_{L}$ and $CF_{U}$, the $N_{S(\vect R)}(\vect T) - CF_{U} + |\vect T \setminus S(\vect R)|$ random indicators that Chernoff tells us must be summed are those with success probabilities $q(Z_i)$ for $i \in \Tilde{\vect J}^{+}_{CF_{U}}$. The second pass of binary search will test values for $h$ in $(CF_{L},CF_{L}+1,...,CF_{U}-1,CF_{U})$ 
and, for each new value $h$, the distributions with success probabilities
$q(Z_i)$ for $i \in \Tilde{\vect J}^{+}_{h} \setminus \Tilde{\vect J}^{+}_{CF_{U}}$ will need to be appended to previous computations to obtain the tail probability.

Let us now ask how we can efficiently compute the initial convolution that Chernoff tells us must contribute to the final result and is quite likely to comprise the bulk of the final result, with the foresight that subsequent iterations will need to reincorporate this bulk into various distinct new tail probability computations until the second pass of binary search halts. We may ask if we can choose a number $L_{div}$ such that the success probabilities $q(Z_i)$ for $i \in \Tilde{\vect J}^{+}_{CF_{U}}$ can be divided into $L_{div}$ groups. In particular, we wish to divide them into $L_{div}$ groups such that each group contains precisely the same mixture of constituent success probabilities. To reiterate once more, we require that each group is an identical mixture of perhaps unequal Bernoulli vectors. If we then calculate the sum of the random indicators within each of $L_{div}$ groups, the result will be $L_{div}$ \textit{iid} distributions as an intermediary step. This strategy has three advantages that easily compensate for the additional overhead of determining the $L_{div}$ groups. First, we only need to compute the within-group sum of indicators once, not $L_{div}$ times, since each group is \textit{iid}. Second, this one convolution needed will contain a small number of random variables of low dimensions to the extent $L_{div}$ is large, and thus, we can exploit DC in the regime it is efficient~\citep{biscarri2018simple} with great payoff. Third, since the $L_{div}$ groups are \textit{iid}, to obtain the distribution for the sum over groups, we can use highly efficient $L_{div}$-fold convolution power, meaning we need only compute one FFT rather than $L_{div}$ FFTs to obtain the final result ($L_{div}$ also cannot be too large to avoid numerical errors in FFT but the algorithm has robust performance for a significant range, say $L_{div} \in [2^2,2^6]$).

Finding $L_{div}$ $iid$ groups is easy and useful to the extent that the initial mixture of success probabilities contains a small number of unique values relative to the total number of indicators. This condition is highly applicable to the scientific context. Clearly, if every initial success probability were unique, it would be impossible to divide them into $L_{div}$ groups giving rise to $L_{div}$ \textit{iid} intermediary sums. Even if there are many repeated success probabilities in the initial mixture, it is unlikely to find a $L_{div}$ in a helpful range to achieve the objective since that requires the frequencies of each unique occurring Bernoulli vector to have $L_{div}$ as a common divisor. In contrast, it is easy to split the initial mixture into $L_{div}$ \textit{iid} groups if we tolerate one group of ``residual'' success probabilities to be triaged and dealt with as a special case. This is the approach we take.

From the probabilities $q(Z_i)$ for $i \in \Tilde{\vect J}^{+}_{CF_{U}}$, compute the unique success probabilities and the frequency of occurrence for each unique vector. Then divide each of the frequencies by $L_{div}$ and round down. The resulting numbers will be how many random indicators of each unique success probability will contribute to one of $L_{div}$ intermediary sums. Rather than using DC to convolve all the Bernoulli vectors in this one group, instead firstly directly compute binomial distributions for the underlying unique groups of \textit{iid} Bernoulli vectors with the number of trials for each binomial as the frequency of the unique occurring Bernoulli vectors divided by $L_{div}$ rounded down. Then, use DC on the resulting binomial vectors; at this stage, DC will still be efficient if $L_{div}$ is large enough. Denote $\vec{p}$ as this resulting probability vector that will be the input to $L_{div}$-fold convolution downstream. In the steps just described, the operation of dividing the frequency of each unique vector in the initial mixture by $L_{div}$ and rounding down produces ``leftovers'' - to become the triaged group - because of rounding. Denote the probability vector of the triaged convolution sum as $\vec{a}$. This concludes the preparation for the second pass; the distributions with success probabilities $q(Z_i)$ for $i \in \Tilde{\vect J}^{+}_{CF_{U}}$ are now encoded in $\vec{a}$, $\vec{p}$, and $L_{div}$.

During the second pass of binary search, $\vec{a}$ is itself immutable but convolved per iteration with distributions that arise from new hypotheses $H_0: \theta_{syn} = h$ producing as a result a temporary vector $\vec{g}$. Specifically, for each new $h$, $\vec{g}$ is obtained by convolving $\vec{a}$ with distributions with success probabilities $q(Z_i)$ for $i \in \Tilde{\vect J}^{+}_{h} \setminus \Tilde{\vect J}^{+}_{CF_{U}}$. These new distributions are convolved with $\vec{a}$ using DC since the dimensions are likely to still be small to the extent Chernoff's bound is close to the exact answer. 

Finally, for each $h$ tested on the second pass, $\vec{p}$ and $\vec{g}$ are sent to an FFT-based step with an exponential tilt to correct for the errors typically induced by FFT. Denote the exponentially tilted versions of $\vec{p}$ and $\vec{g}$ as $\vec{p}_{\hat{s}}$ and $\vec{g}_{\hat{s}}$, respectively. Specifically, the convolution for each $h$ occurs via $L_{div}$-fold convolution power for $\vec{p}_{\hat{s}}$ and point-wise multiplication of the result with $\vec{g}_{\hat{s}}$ in the frequency domain (each embedded in the dimension of the final result). It is important to mention that one might argue we only need to compute the FFT of $\vec{p}_{\hat{s}}$ once after the first pass, and hence $\vec{p}_{\hat{s}}$ can then sit idly in the frequency domain during the second pass waiting to be point-wise multiplied with new versions of $\vec{g}_{\hat{s}}$ for each new test $h$ encountered by the second pass of binary search. However, this does not work because every new convolution requires a new exponential tilting parameter. That is, $\vec{p}_{\hat{s}}$ is a temporary vector that changes each iteration because a new tilting parameter is required. While this prohibits such reuse~\citep{wilson2016accurate}, it offers extraordinarily large increases in accuracy in return. We prioritize accuracy over this decrease in speed at this juncture. Overall, the algorithm is still competitively fast.

The iterative FFT convolution with exponential tilting is detailed and given its own description in Algorithm~\ref{alg:convolution}. The overall procedure is then summarised in Algorithm~\ref{alg:confidence} for an excitatory lower confidence bound. There are three more cases: the excitatory upper bound as well as inhibitory lower and upper bounds. The other cases are very similar but we quickly remark here on the differences. First, a lower bound of $\theta_{syn}$ corresponds to the right tail probability of the test statistic whereas an upper bound for $\theta_{syn}$ corresponds to a left tail probability of the test statistic. Exponential tilting is suitable for computing right tail probabilities, but for any random variable $X$ the left tail probability can be obtained by computing the right tail probability of $-X$. Last, the initial implementation of binary search using Chernoff's bound for excitatory intervals searches $h$ on the sequence $(1,2,...,N_{S(\vect R)}(\vect T))$ but for inhibition the sequence searched is $(1,2,...,|\vect G|)$.

Algorithm~\ref{alg:confidence} begins with the assumption that the null hypothesis $H_0: \theta_{syn} = 0$ has already been rejected and code posted online tests $H_0: \theta_{syn} = 0$ and breaks if it fails to be rejected. The philosophy behind this decision is that in large-scale neural recordings, most neurons will be unconnected, and computing confidence intervals for all pairwise interactions will be expensive and not insightful. A very rapid approximation, such as the normal approximation, can be used for the test $H_0: \theta_{syn} = 0$, and a p-value can be calibrated for the precision a specific question requires. The confidence intervals can then be applied to significant pairs and still may include $0$ when the exact method is employed. Code posted online also gives the option to return confidence intervals computed with p-values obtained from the normal approximation which is much faster for large spike trains. Preliminary numerical experiments suggest spike trains of moderate size will yield normal approximations very close to the exact solution. However, the proximity to the exact solution depends upon many factors and the tolerance for error depends on the scientific question, necessitating specific use case evaluations.

\subsection{Simulation software and hardware}
To ensure careful handling of frozen noise inputs, all numerical simulations of dynamical systems were programmed from scratch and integrated with modified Euler's method in Python. Long simulations were optimized often through parallelization and concatenation. This was done in Google Colab using the NVIDIA T4 GPU.

\section*{Acknowledgements}
We thank Jonathan Platkiewicz, Gyorgy Buzsaki, Daniel English, Uri Keich, and Thibaud Taillefumier for helpful conversations that guided this research.

\let\oldbibitem\bibitem
\renewcommand{\bibitem}[2][]{\oldbibitem[#1]{#2}\vspace{0.5ex}}
\addcontentsline{toc}{section}{References} 
\bibliography{monosynapse}

\end{document}